\documentclass[letterpaper,11pt]{article}
\usepackage{fullpage}
\usepackage{booktabs} 
\usepackage[ruled]{algorithm2e} 

\SetAlFnt{\small}
\SetAlCapFnt{\small}
\SetAlCapNameFnt{\small}
\SetAlCapHSkip{0pt}
\IncMargin{-\parindent}

\usepackage[numbers]{natbib}
\usepackage{amsmath,amsfonts,amsthm,amsbsy,amssymb,bm}
\usepackage[hidelinks,colorlinks,allcolors=magenta]{hyperref}
\usepackage{mathtools}
\usepackage[thinc]{esdiff}
\usepackage{enumitem}  
\usepackage{multirow}
\usepackage{nicefrac}
\usepackage{stackrel}
\usepackage{subcaption}
\usepackage{graphicx}
\usepackage{xcolor}
\usepackage{tikz}
\usepackage{tikz}
\usetikzlibrary{arrows.meta}
\usetikzlibrary{matrix,shapes,arrows,fit}

\usepackage[nameinlink]{cleveref}

\renewcommand{\paragraph}[1]{\medskip\noindent\textbf{#1}~}

\newcommand{\N}{\mathbb{N}}
\newcommand{\R}{\mathbb{R}}
\renewcommand{\ge}{\geqslant}

\renewcommand{\le}{\leqslant}

\DeclareMathOperator*{\NW}{NW}

\theoremstyle{plain}
\newtheorem{theorem}{Theorem}

\newtheorem{proposition}{Proposition}
\newtheorem{corollary}{Corollary}

\theoremstyle{definition}
\newtheorem{definition}{Definition}
\newtheorem{example}{Example}
\newtheorem{remark}{Remark}

\DeclareMathOperator*{\argmax}{\arg\max}

\newcommand{\dd}{\mathrm{d}}

\renewcommand{\O}{\ensuremath{\mathcal{O}}}
\newcommand{\U}{\ensuremath{\mathcal{U}}}
\renewcommand{\vec}{\bm}
\newcommand{\xx}{\vec{x}}
\newcommand{\yy}{\vec{y}}
\newcommand{\zz}{\vec{z}}
\newcommand{\uu}{\vec{u}}

\newcommand{\peffull}{individual harm ratio\xspace}
\newcommand{\pgeffull}{equal-sized group harm ratio\xspace}
\newcommand{\pgffull}{group harm ratio\xspace}
\newcommand{\pef}{IHR\xspace}
\newcommand{\pgef}{EGHR\xspace}
\newcommand{\pgf}{GHR\xspace}

\newcommand*\circled[1]{\tikz[baseline=(char.base)]{
            \node[shape=circle,draw,inner sep=2pt] (char) {$#1$};}}

\newcommand{\agents}{\ensuremath{N}}

\newcommand{\outcomes}{\ensuremath{\mathcal{C}}}

\begin{document}\allowdisplaybreaks
\title{Harm Ratio: A Novel and Versatile Fairness Criterion\thanks{A preliminary version was published in the proceedings of the 4th ACM Conference on Equity and Access in Algorithms, Mechanisms, and Optimization (EAAMO), 2024.}}

\author{
Soroush Ebadian\textsuperscript{\rm 1} 
\and 
Rupert Freeman\textsuperscript{\rm 2}
\and
Nisarg Shah\textsuperscript{\rm 1}
}
\date{
\textsuperscript{\rm 1}University of Toronto, \texttt{\{soroush,nisarg\}@cs.toronto.edu},\\
\textsuperscript{\rm 2}University of Virginia, \texttt{FreemanR@darden.virginia.edu}
}

\maketitle
\begin{abstract}
Envy-freeness has become the cornerstone of fair division research. In settings where each individual is allocated a disjoint share of collective resources, it is a compelling fairness axiom which demands that no individual strictly prefer the allocation of another individual to their own. Unfortunately, in many real-life collective decision-making problems, the goal is to choose a (common) public outcome that is equally applicable to all individuals, and the notion of envy becomes vacuous. Consequently, this literature has avoided studying fairness criteria that focus on individuals feeling a sense of jealousy or resentment towards other individuals (rather than towards the system), missing out on a key aspect of fairness. 

In this work, we propose a novel fairness criterion, \emph{\peffull}, which is inspired by envy-freeness but applies to a broad range of collective decision-making settings. Theoretically, we identify minimal conditions under which this criterion and its groupwise extensions can be guaranteed, and study the computational complexity of related problems. Empirically, we conduct experiments with real data to show that our fairness criterion is powerful enough to differentiate between prominent decision-making algorithms for a range of tasks from voting and fair division to participatory budgeting and peer review.  
\end{abstract}

\section{Introduction}

How to make collective decisions while treating (groups of) individuals in a fair manner is a question that human societies have struggled to understand for centuries. Today, with the advent of algorithms making increasingly critical decisions, the field of algorithmic fairness has considered a plethora of novel fairness criteria~\cite{pessach2023algorithmic}, albeit they often have limited applicability due to being handcrafted for specialized decision-making tasks. To surpass this limitation, researchers have looked towards computational social choice, a field at the intersection of economics and computer science~\cite{finocchiaro2021bridging,gummadi2019economic}, where fairness criteria applicable to a broad range of domains have been proposed. 

\emph{Envy-freeness}~\cite{Tinb30,GS58,Fol67} is arguably the most well-studied fairness criterion in computational social choice, defined for a general resource allocation setting in which a set of goods $M$ is to be divided between a set of agents $N = \{1,\ldots,n\}$. The goal is to find an allocation $A = (A_1,\ldots,A_n)$, which is a partition of $M$, with $A_i$ denoting the allocation to agent $i$. Each agent $i$ has a utility function $u_i : 2^M \to \R_{\ge 0}$ and her utility under allocation $A$ is given by $u_i(A_i)$. Envy-freeness demands that no agent strictly prefer the allocation to another agent over her own, i.e., $u_i(A_i) \ge u_i(A_j)$ for all $i,j \in N$. Envy-freeness is the canonical embodiment of two interesting features of a fairness criterion: \begin{enumerate}
    \item It is a pairwise individual fairness criterion, where an agent $i$, who feels that she is being treated unfairly, points to another agent $j$ as the reason for it.
    \item It makes no interpersonal comparisons, i.e., the utility function $u_i$ of agent $i$ is never compared to the utility function of any other agent.
\end{enumerate}
Significant literature has been devoted to exploring envy-freeness and its relaxations for a variety of resource allocation domains beyond the allocation of homogeneous divisible goods~\cite{Pro16,ABFV22}. 

A problem arises when one moves to more general collective decision-making domains. In the most general model, which we term the \emph{public outcomes model}, there is a set $N$ of $n$ agents and a set $\O$ of possible outcomes. Each agent $i$ has a utility function $u_i : \O \to \R_{\ge 0}$ and the goal is to choose an outcome $o \in \O$. \emph{What would envy-freeness mean here?} Even conceptually, can an agent $i$ really envy another agent $j$, despite both experiencing the common outcome $o$? From this vantage point, envy-freeness may seem vacuous for public outcomes. 

Due to this, the literature on (various special cases of) the public outcomes model has ignored envy-freeness, which misses an essential consideration where an individual feels they have been unfairly treated due to another individual receiving disproportionate importance.

Note that the public outcomes model captures the resource allocation model as a special case, by setting $\O$ as the set of all possible allocations $A$ and defining $u_i(A) \triangleq u_i(A_i)$. But importantly, it also captures a host of other collective decision-making problems such as choosing the winner of an election based on voters’ preferences, choosing a company's strategy based on its employees' opinions, deciding which graduate applicants to admit based on the preferences of faculty members, or picking interesting posts to showcase on a social media website's (non-personalized) front page based on its users' interests. 

\subsection{Our Contributions}\label{sec:contributions}

In \Cref{sec:new-def-relations}, we propose our novel pairwise individual fairness criterion, \emph{\peffull} (\pef), which is inspired by envy-freeness, makes no interpersonal comparisons, and applies to the full range of the public outcomes model. We observe that in resource allocation with additive utilities, 1-\pef logically implies (and is significantly stronger than) envy-freeness. We also define its groupwise extensions, \emph{\pgeffull} (\pgef) and \emph{\pgffull} (\pgf); analyze their relations to each other as well as to analogous extensions of envy-freeness in resource allocation~\cite{Var74,BTD92,CFSW19}; and reconstruct a hierarchy of fairness criteria from resource allocation in the broader public outcomes model (\Cref{fig:hierarchy}). 

In \Cref{sec:mnw}, we show that even a \pgffull (\pgf) of $1$ (which implies 1-\pef) can be guaranteed in the public outcomes model by maximizing the Nash welfare (geometric mean of agent utilities), when the set of feasible utility vectors $\U = \{(u_1(o),\ldots,u_n(o)) : o \in \O\}$ satisfies two mild conditions of compactness and upper convexity. In doing so, we establish an equivalence to proportional fairness and generalize several prior results; see \Cref{sec:significance_summary}. We also prove that any outcome with a \pgffull of $1$ achieves a (tight) $\binom{n}{\lfloor{n/2}\rfloor}^{\nicefrac{1}{n}}$-approximation --- at least as good as a $2$-approximation --- of the maximum Nash welfare in any instance of the public outcomes model. This complements an existing characterization of maximum Nash welfare via group fairness due to \citet{FSV20}.  

In \Cref{sec:computation}, we observe that in a special case of the public outcomes model, a $(1+\epsilon)$-\pgf outcome can be computed in time polynomial in the instance size and $1/\epsilon$ using convex programming, but leave polynomial-time computation of a 1-\pgf (or even 1-\pef) outcome as an open question. For the allocation of indivisible goods under additive utilities, given an allocation, while checking whether it is EF is trivially in P, we show that checking whether it is 1-\pef (which implies EF) is coNP-complete.

In \Cref{sec:experiments}, we conduct experiments with simulated and real datasets for the allocation of indivisible goods (\href{www.spliddit.org}{Spliddit}), peer review (CVPR 2017, CVPR 2018, ICLR 2018), and participatory budgeting (\href{www.pabulib.org}{Pabulib}). Our findings suggest that individual and group harm ratios are powerful enough to distinguish between prominent rules in terms of the level of fairness they achieve, and correlate well with existing domain-specific fairness notions in each domain. 

\subsection{Summary of Significance and Implications}\label{sec:significance_summary}

We argue the significance of our work in terms of providing a unifying framework to think of fairness in a broad range of settings and opening the door to novel styles of fairness analyses in well-studied domains. In \Cref{app:significance}, we provide a technical comparison of our results to those from prior works, and identify exactly which prior results our work subsumes. Here, we briefly summarize four significant implications of our work. 

\begin{enumerate}
    \item Our work subsumes existing results from the resource allocation literature proving that MNW allocations are envy-free (EF) and Pareto optimal (PO), e.g., in cake-cutting~\cite{Wel85,SS19} (with our proof being more elementary than the existing ones), or that they are approximately EF and PO, e.g., in one-sided matching~\cite{troebst2024cardinal}. This is a result of $1$-\pef implying exact or approximate EF in these domains. 
    \item Our work adds to a long line of literature identifying conditions guaranteeing the existence of EF+PO allocations~\cite{Var74,Sven83,diamantaras1992equity} and outcomes satisfying another fairness criterion, proportional fairness~\cite{ray2022fairness}.
    \item In some domains where our results hold, the fact that MNW allocations satisfy $1$-\pef and PO is novel and interesting~\cite{Dall01,cole2021existence,caragiannis2022beyond}; in one case, we are able to establish (a previously unknown) existence of EF+PO allocations~\cite{Dall01}.
    \item Finally, our work significantly expands the possibility of conducting pairwise individual fairness analyses to a much larger set of collective decision-making domains than the resource allocation domains where envy-freeness is well-defined.
\end{enumerate}

\section{Model}
For $t \in \N$, define $[t] \triangleq \{1,2,\ldots,t\}$. We consider a very general multi-agent decision-making setting, which we term the \emph{public outcomes} model, where each instance is given by a set of outcomes $\O$, a finite set of agents $N=[n]$, and a utility profile $\uu = (u_1,\ldots,u_n)$ containing a utility function $u_i : \O \to \R_{\ge 0}$ for each agent $i \in N$. Without loss of generality, we assume that for each agent $i \in N$ there is some outcome $o \in \O$ such that $u_i(o) > 0$; otherwise, agent $i$ can effectively be removed from consideration. 

A priori, we impose no assumption on the set of outcomes $\O$ (it could be finite, countably infinite, or uncountable) or the utility functions of the agents. As such, this general model subsumes various models studied in the literature as special cases, including cake-cutting~\cite{Stein48}, allocation of homogeneous divisible or indivisible goods~\cite{Moul04}, public decision-making~\cite{CFS17}, and allocation of public goods~\cite{FMS18}.

\textbf{Utility set.} Define the \emph{utility set} $\U = \{(u_1(o),\ldots,u_n(o)) : o \in \O\}$ to be the set of feasible utility vectors. Because we have placed no assumptions on the agent utilities, this is essentially the only object of interest as, given any $\U$, one can easily construct underlying outcome set $\O$ and utility profile $\uu$ which induce $\U$. Crucially, the fairness notions we define below also depend only on $\U$.

\subsection{Fairness}\label{sec:prelim-fairness}
While many fairness notions defined in the literature for specialized models do not extend to this general model, the following ones do, and play a key role in our work. 
First, an outcome is proportionally fair if the agents are not happier (in an average multiplicative sense) when switching to any other outcome.

\begin{definition}[Proportional Fairness (PF)]\label{def:pf}
    We say that an outcome $o \in \O$ is proportionally fair (PF) if $\frac{1}{n} \sum_{i \in N} \frac{u_i(o')}{u_i(o)} \le 1$ for all $o' \in \O$.\footnote{If $u_i(o) = 0$, we let $u_i(o')/u_i(o) = \infty$ and PF is automatically violated.} 
\end{definition}

A maximum Nash welfare outcome is one that maximizes the geometric mean of agent utilities.

\begin{definition}[Maximum Nash Welfare (MNW)]\label{def:mnw}
Define the Nash welfare of an outcome $o \in \O$ as the geometric mean of agent utilities for it: $\NW(o) = (\prod_{i \in N} u_i(o))^{\nicefrac{1}{n}}$.\footnote{If $u_i(o)=0$ for any agent $i \in N$, we say that $\NW(o) = 0$.} For $\alpha \ge 1$, 
we say that an outcome $o \in \O$ is an $\alpha$-approximate maximum Nash welfare (MNW) outcome if $\NW(o) \ge (1/\alpha) \cdot \NW(o')$ for all $o' \in \O$. When $\alpha = 1$, we simply refer to it as an MNW outcome. 
\end{definition}

An outcome is in the core~\cite{Fol67} if no group can find an alternative outcome that provides all of them weakly higher utility, and at least one of them strictly higher utility, even after scaling their utilities by the size of the group in proportion to the set of all agents.

\begin{definition}[Core]\label{def:core}
    We say that an outcome $o \in \O$ is in the core if there is no group of agents $S \subseteq N$ and outcome $o' \in \O$ such that $\frac{|S|}{n} \cdot u_i(o') \ge u_i(o)$ for all $i \in S$ and at least one inequality is strict. 
\end{definition}

Proportionality~\cite{Stein48} says that each agent should receive at least a $\nicefrac{1}{n}$ fraction of the utility she can achieve in any outcome. 

\begin{definition}[Proportionality (Prop)]\label{def:prop}
    We say that an outcome $o \in \O$ is proportional (Prop) if $u_i(o) \ge \frac{1}{n} u_i(o')$ for all $o' \in \O$ and all $i \in N$. 
\end{definition}

Finally, the following is traditionally considered a notion of efficiency, but we include it here because it connects well to the notions defined above. An outcome is Pareto optimal if there is no alternative outcome that makes each agent at least as happy and some agent strictly happier. 

\begin{definition}[Pareto Optimality (PO)]\label{def:po}
    We say that an outcome $o \in \O$ is Pareto optimal (PO) if there is no outcome $o' \in \O$ such that $u_i(o') \ge u_i(o)$ for all $i \in N$ and at least one inequality is strict. 
\end{definition}

Note that these notions are \emph{utilitarian}, i.e., they depend only on the utility vector induced by the outcome. Hence, one can equivalently speak of a utility vector satisfying these notions; if a utility vector satisfies a utilitarian notion, every outcome inducing it must satisfy the notion as well. 

The following proposition describes well-known relationships between these notions.

\begin{proposition}[PF $\Rightarrow$ MNW, PF $\Rightarrow$ Core $\Rightarrow$ (Prop+PO)]\label{prop:implications-prelim}
    Any proportionally fair outcome is also a maximum Nash welfare outcome and lies in the core. Any outcome in the core satisfies proportionality and Pareto optimality.
\end{proposition}

\begin{proof}
\Cref{prop:implications-prelim}]
    \emph{PF $\Rightarrow$ MNW.} Suppose $o \in \O$ is PF. Then, $u_i(o) > 0$ for all $i \in N$. For any $o' \in \O$, we have
    \[
    \left(\prod_{i \in N} \frac{u_i(o')}{u_i(o)}\right)^{{1/n}} \le \frac{1}{n} \sum_{i \in N} \frac{u_i(o')}{u_i(o)} \le 1,
    \]
    where we use the AM-GM inequality followed by the definition of PF. This shows that $\NW(o') \le \NW(o)$, as desired.

    \emph{PF $\Rightarrow$ Core.} Suppose $o \in \O$ is PF but violates the core. Then, there exists a group of agents $S \subseteq N$ and an outcome $o' \in \O$ such that $\frac{u_i(o')}{u_i(o)} \ge \frac{n}{|S|}$ for all $i \in S$ and at least one inequality is strict (we can place $u_i(o)$ in the denominator because $o$ being PF ensures $u_i(o) > 0$). Thus, we have
    \[
    \frac{1}{n} \cdot \sum_{i \in N} \frac{u_i(o')}{u_i(o)} \ge \frac{1}{n} \cdot \sum_{i \in S} \frac{u_i(o')}{u_i(o)} > \frac{1}{n} \cdot |S| \cdot \frac{n}{|S|} = 1,
    \]
    which violates $o$ being PF. 

    \emph{Core $\Rightarrow$ (Prop+PO).} Core $\Rightarrow$ Prop follows from observing that proportionality imposes the same constraint as the core, but only for groups of agents $S$ with $|S|=1$. Similarly, Pareto optimality imposes the same constraint as the core but only for $S = N$. 
\end{proof}

\begin{remark}
In many domains, every agent can possibly get zero utility under some outcome (i.e., $\forall i, \exists o \in \O : u_i(o) = 0$). However, when this is not the case, one can define stronger versions of the core and proportionality that take $u_i^{\min} \triangleq \inf_{o \in \O} u_i(o)$ into account. 

In the definition of the core, we would write $\frac{|S|}{n} (u_i(o')-u_i^{\min}) \ge (u_i(o)-u_i^{\min})$, and for proportionality, we would write $u_i(o) \ge \frac{1}{n} u_i(o') + \frac{n-1}{n} u_i^{\min} \Leftrightarrow u_i(o)-u_i^{\min} \ge \frac{1}{n} (u_i(o')-u_i^{\min})$. One can check that this strengthens the respective definition. \citet{ASSW23} refer to the strengthened proportionality as general fair share (GFS).

Any method of achieving the regular versions can be modified to achieve the strengthened versions by feeding it translated utility functions $\{u'_i\}_{i \in N}$, where $u'_i(o) = u_i(o)-u_i^{\min}$ for all $i \in N$ and $o \in \O$. For maximum Nash welfare, we would then seek an outcome $o$ that maximizes $\prod_{i \in N} (u_i(o)-u_i^{\min})$; this is akin to the Nash bargaining solution~\cite{Nash50b}, if $u_1^{\min},\ldots,u_n^{\min}$ are treated as status quo utilities. Similarly, for proportional fairness, we would seek an outcome $o$ for which $\frac{1}{n} \sum_{i \in N} \frac{u_i(o')-u_i^{\min}}{u_i(o)-u_i^{\min}} \le 1$. 
\end{remark}

\subsection{Private Goods Division}\label{sec:prelim-private}
As demonstrated in the introduction, the public outcomes model subsumes a wide range of collective decision-making models studied in the literature. Below, we formally introduce the special case of private goods division, which we will refer to frequently in our results. 

In a general \emph{private goods division} setting, there is a set of goods $M$ to be divided between a set of agents $N = [n]$. The outcome set $\O$ is the set of allocations $A = (A_1,\ldots,A_n)$ which are partitions of $M$ into pairwise-disjoint measurable bundles. Each agent $i \in N$ has a measure $u_i$ over $M$.\footnote{A measure assigns a non-negative value to each measurable subset of $M$ ($0$ to the empty subset) and is countably additive.} Crucially, her utility for an allocation $A$ is solely a function of the bundle $A_i$ assigned to her, namely $u_i(A_i)$. 

This model captures three prominent settings studied in the literature:
\begin{enumerate}
    \item \emph{Cake-cutting:} $M = [0,1]$ and the utility function $u_i$ of each agent $i$ is a (countably additive) measure over $M$ that is absolutely continuous with respect to the Lebesgue measure.
    \item \emph{Homogeneous divisible goods:} $M$ is a set of $m$ homogeneous divisible goods. An allocation can be described as $A = (A_{i,g})_{i \in N, g \in M}$ with $A_{i,g}$ being the fraction of good $g$ allocated to agent $i$ and $\sum_{i \in N} A_{i,g} = 1$ for all $g$. The utility of each agent $i$ can be given by $u_i(A_i) = \sum_{g \in M} A_{i,g} \cdot v_{i,g}$, where $v_{i,g}$ is her value for receiving good $g$ entirely. 
    \item \emph{Indivisible goods:} This is identical to homogeneous divisible goods, except we further restrict $A_{i,g} \in \{0,1\}$ for all $i \in N$ and $g \in M$. 
\end{enumerate}

For private goods division, additional fairness notions have been extensively studied, but they do not extend to the more general model of public outcomes; coming up with natural extensions of these notions is precisely the subject of our work. 
Envy-freeness~\cite{Fol67} says that no agent should prefer the allocation of another agent to her own. 

\begin{definition}[Envy-Freeness (EF)]\label{def:ef}
    We say that an allocation $A$ is envy-free (EF) if $u_i(A_i) \ge u_i(A_j)$ for all $i,j \in N$.
\end{definition}

Group fairness~\cite{CFSW19} requires that no group of agents should ``envy'' (in the sense of being able to Pareto improve by taking the resources allocated to the other group and redistributing amongst themselves, subject to an appropriate scaling factor to account for different-sized groups) any other group of agents.

\begin{definition}[Group Fairness (GF)]\label{def:gf}
    We say that an allocation $A$ is group fair (GF) if there is no pair of groups of agents $S,T \subseteq N$ and a division $B$ of $\cup_{j \in T} A_j$ among agents in $S$ such that $\frac{|S|}{|T|} \cdot u_i(B_i) \ge u_i(A_i)$ for all $i \in S$ and at least one inequality is strict.
\end{definition}

For private goods division, envy-freeness implies proportionality, and group fairness implies the core (and therefore also proportionality). Additionally, restricting the group fairness definition to allow only pairs of groups with $|S|=|T|$ yields the group envy-freeness definition of~\citet{BTD92}.

Note that the difficulty in extending these notions is that in the public outcomes model, there is nothing ``allocated'' to individual agents. Instead, a common outcome is selected for all agents, making the concept of envy between individuals or groups vacuous. 

\section{Harm Ratio: A Novel Fairness Criterion}\label{sec:new-def-relations}

To motivate the definition of \peffull, let us first revisit envy-freeness (EF) in private goods division with additive utilities (Definition~\ref{def:ef}): an allocation $A$ is EF if $u_i(A_i) \ge u_i(A_j)$ for all $i,j \in \agents$. This is not a utilitarian notion: envy-freeness cannot be checked simply from the induced utility vector $(u_1(A),\ldots,u_n(A))$, which makes it difficult to extend to the public outcomes model. However, if $A$ violates EF, there exists a new allocation $A'$, given by $A'_i = A_i \cup A_j$, $A'_j = \emptyset$, and $A'_k = A_k$ for all $k \in N\setminus\{i,j\}$, which induces a utility vector $(u_1(A'),\ldots,u_n(A'))$ satisfying the following: $u_i(A') > 2 \cdot u_i(A)$ and $u_k(A') \ge u_k(A)$ for all $k \in N\setminus\{i,j\}$. This is a purely utilitarian comparison, which can be generalized well beyond private goods division to the public outcomes model. As argued in the introduction, the utility improvement factor of $2$ in this case measures the level of harm imposed on agent $i$ due to the presence of agent $j$. We refer to it as the \peffull (\pef).

The factor of $2$ is important. Not only does it show up in connection to private goods division with additive utilities, we prove in Theorem~\ref{thm:mnw-pf} that it remains the tightest achievable factor for the more general public outcomes model under certain conditions on the utility set $\U$. It also generalizes well to a factor of $\frac{|S \cup T|}{|S|}$ in the groupwise extension of \pef introduced later (Definition~\ref{def:pgf}). However, a worse factor may be achieved when the utilities do not satisfy the requirements of Theorem~\ref{thm:mnw-pf}, or a better factor may be achieved in practice on real-world instances. To account for these possibilities, we introduce an $\alpha$-approximation of \pef.

\begin{definition}[Individual Harm Ratio]\label{def:pub-ef}
    We say that an outcome $o \in \O$ has an \emph{\peffull} of $\alpha$, denoted $\alpha$-\pef, if there are no agents $i,j \in N$ and outcome $o' \in \O$ such that $\frac{1}{2} \cdot u_i(o') > \alpha \cdot u_i(o)$ and $u_k(o') \ge u_k(o)$ for all $k \in N \setminus \{i,j\}$.
\end{definition}

Definition~\ref{def:pub-ef} requires that no agent be able to find an outcome in which even half of her utility is higher (by a factor of $\alpha$) and at most one other agent is hurt. Intuitively, a large \peffull gives agent $i$ a ``justified claim'' that the system is treating her unfairly due to its attempt to make agent $j$ happy. The higher the $\alpha$, the stronger the claim. Conversely, the lower the value of $\alpha$ for an outcome, the stronger its fairness guarantee.

\subsection{Comparison to Envy-Freeness and Proportionality}\label{sec:pef-private}

Consider private goods division with additive utilities. From the argument above, it is clear that $1$-\pef implies envy-freeness, but in fact, it is strictly stronger: even if moving the goods allocated to agent $j$ to agent $i$ does not suffice to more than double agent $i$'s utility, a complete reshuffling that also alters the allocations to the other agents may nonetheless be able to achieve this, while still keeping the other agents at least as happy. The next result shows that envy-freeness is much weaker than $1$-\pef, as it implies only $(\nicefrac{n}{2})$-\pef.

\begin{figure}[t]
    \centering
    \begin{tikzpicture}
      \matrix (M) [%
        matrix of nodes, column sep=3pt, row sep=3pt
      ]
      {%
        { }& $g_1$ & $g_2$ & $g_3$ & $\ldots$ & $g_{n - 1}$ & $g_n$\\
         $u_1$ & $1$ & $0$ & $0$ & $\ldots$ & $0$ & $0$ \\
         $u_2$ & $1$ & $1$ & $0$ & $\ldots$ & $0$ & $0$ \\
         $u_3$ & $1$ & $1$ & $1$ & $\ldots$ & $0$ & $0$ \\
         $\vdots$ &&&& $\ddots$ && $\vdots$ \\
         $u_{n-1}$ & $1$ & $1$ & $1$ & $\ldots$ & $1$ & $0$ \\
         $u_n$ & $1$ & $1$ & $1$ & $\ldots$ & $1$ & $1$ \\
       };
       \node[draw=red,rounded corners = 1ex,fit=(M-2-2),inner sep = 0pt] {};
       \node[draw=red,rounded corners = 1ex,fit=(M-4-3)(M-7-3),inner sep = 1pt] {};
       \node[draw=blue,rounded corners = 1ex,fit=(M-2-2)(M-7-2),inner sep = 2pt] {};
       \node[draw=blue,rounded corners = 1ex,fit=(M-3-3),inner sep = 1pt] {};
       \node[draw=blue,rounded corners = 1ex,fit=(M-4-4),inner sep = 2pt] {};
       \node[draw=red,rounded corners = 1ex,fit=(M-4-4),inner sep = 0pt] {};
       \node[draw=blue,rounded corners = 1ex,fit=(M-6-6),inner sep = 2pt] {};
       \node[draw=red,rounded corners = 1ex,fit=(M-6-6),inner sep = 0pt] {};
       \node[draw=blue,rounded corners = 1ex,fit=(M-7-7),inner sep = 2pt] {};
       \node[draw=red,rounded corners = 1ex,fit=(M-7-7),inner sep = 0pt] {};
    \end{tikzpicture}
    \caption{Instance for \Cref{thm:pef-ef}. Blue and red rectangles show two different allocations; the rectangle in each good's column covers the agents among which the good is equally divided. The blue allocation is EF, but $(\nicefrac{n}{2})$-\pef, as witnessed by the existence of the red allocation.}
    \label{tab:ef-vs-publicef}
\end{figure}

\begin{theorem}\label{thm:pef-ef}
For private goods division with additive utilities:
\begin{itemize}
    \item {\normalfont ($1$-\pef $\Rightarrow$ EF)} An \peffull of $1$ implies envy-freeness.
    \item {\normalfont (Prop $\Rightarrow$ $(\nicefrac{n}{2})$-\pef)} Proportionality (and therefore envy-freeness) implies an \peffull of $\nicefrac{n}{2}$, and this is tight even for the allocation of homogeneous divisible goods.
\end{itemize}
\end{theorem}

\begin{proof}
Fix any private goods division instance with additive utilities. We have already argued that $1$-\pef implies EF. To see that Prop implies $(\nicefrac{n}{2})$-\pef, consider any Prop allocation $A$. Hence, $u_i(A_i) \ge \frac{1}{n} \cdot u_i(M)$. Suppose for contradiction that $A$ is not $(\nicefrac{n}{2})$-\pef. Hence, there exist agents $i,j \in N$ and allocation $A'$ such that $u_i(A'_i) > n \cdot u_i(A_i) \ge u_i(M)$, which is a contradiction. 

For $n=2$, this shows that $1$-\pef is equivalent to EF, which is also evident from the definitions of \pef and EF. To show tightness for $n \ge 3$, consider the instance in \Cref{tab:ef-vs-publicef} where $n$ homogeneous divisible goods need to be divided between $n$ agents. Agent $i \in [n]$ values items $j \in \{1, \ldots, i\}$ at $1$ and the rest at $0$. 

Consider an allocation $A$ (shown via blue rectangles in \Cref{tab:ef-vs-publicef}) that divides $g_1$ equally among all agents and allocates good $j \in [n] \setminus \{1\}$ to agent $j$. Then, $u_1(A_1) = \frac{1}{n}$ and $u_i(A_i) = 1 + \frac{1}{n}$ for all $i\in[n] \setminus \{1\}$. It is easy to check that this is EF. However, we show that agent $1$ is significantly harmed by the presence of agent $2$ in the sense of \pef. Consider the allocation $A'$ (shown via red rectangles in \Cref{tab:ef-vs-publicef}) where each agent $i \in [n] \setminus \{2\}$ is allocated the $i$-th good, and the second good is equally divided among agents in $\{3,\ldots,n\}$. Then, $u_1(A'_1) = 1$ and $u_i(A'_i) = 1 + \frac{1}{n - 2} \ge 1 + \frac{1}{n} = u_i(A_i)$ for all $i \in [3, n]$. Each agent $i \in N \setminus \{1,2\}$ is at least as happy under $A'$ as under $A$, and for agent $1$, $u_1(A'_1) = 1 = \frac{n}{2} \cdot 2 \cdot u_1(A_1)$. Hence, $A$ is only $(\nicefrac{n}{2})$-\pef.
\end{proof}

What happens once we go beyond private goods division? As noted earlier, envy-freeness is no longer well-defined for the public outcomes model, but proportionality still is. For private goods division with additive utilities, we noticed above that $1$-\pef implies proportionality. This breaks down when we move to the public outcomes model, as the following simple example shows. 

\begin{example}\label{ex:pef-weak}
    Consider an instance of the public outcomes model with a set of two outcomes $\O = \{o,o'\}$ and a set of three agents $N = [3]$. The utilities are $(u_1(o),u_2(o),u_3(o)) = (1,1,0)$ and $(u_1(o'),u_2(o'),u_3(o')) = (0.1,0.1,1)$. Outcome $o$ is $1$-\pef: agents $1$ and $2$ are maximally happy, and while agent $3$ receives zero utility, she cannot choose any other outcome where only one agent is hurt and thus cannot improve at all in any outcome that hurts at most one other agent. Note that $o$ gives a zero approximation of proportionality for agent $3$. 
\end{example}

\subsection{Groupwise Extension}

Example~\ref{ex:pef-weak} points out a weakness of \pef in the public outcomes model. Finding an outcome that hurts at most one other agent can be quite limiting if there are no such alternative outcomes. What if we allow an agent (or a group of agents) to find outcomes that hurt multiple other agents? This must be accompanied with a greater requirement of utility improvement. It turns out we can strengthen \pef to \pgffull (\pgf) in the same way that envy-freeness is strengthened to group fairness (GF) in private goods division~\cite{CFSW19}. 

\begin{definition}[Group Harm Ratio]\label{def:pgf}
We say that an
outcome $o \in \O$ has a \emph{\pgffull} of $\alpha$, denoted $\alpha$-\pgf, if there are no non-empty groups of agents $S,T \subseteq N$ and outcome $o' \in \O$ such that the following two conditions hold:
\begin{enumerate}
    \item For all agents $i \in S$, $\frac{|S|}{|S \cup T|} \cdot u_i(o') \ge \alpha \cdot u_i(o)$ with at least one inequality being strict.
    \item For all agents $i \in N\setminus (S \cup T)$, $u_i(o') \ge u_i(o)$.
\end{enumerate}
\end{definition}

We remark that one can also define \pgeffull (\pgef) by only imposing the \pgf constraints when $|S|=|T|$, the same way GEF is defined for private goods division. 

\begin{proposition}\label{prop:implications-public}
In the public outcomes model, PF $\Rightarrow$ $1$-\pgf $\Rightarrow$ (PO + Core + $1$-\pef).
\end{proposition}
\begin{proof}
PF $\Rightarrow$ $1$-\pgf. Suppose $o \in \O$ is PF but violates $1$-\pgf. Then there exist non-empty $S, T \subseteq N$ and outcome $o' \in \O$ such that $\frac{|S|}{|S \cup T|} \cdot u_i(o') \ge u_i(o)$ for all $i \in S$, with at least one inequality being strict, and $u_i(o') \ge u_i(o)$ for all $i \in N \setminus (S \cup T)$. If $u_i(o) = 0$ for any $i \in N$, then $o$ clearly violates PF, which is the desired contradiction. 

Assume that $u_i(o) > 0$ for all $i \in N$. Then we have
\begin{align*}
    \frac{1}{n} \sum_{i \in N} \frac{u_i(o')}{u_i(o)} &\ge \frac{1}{n} \left( \sum_{i \in S} \frac{u_i(o')}{u_i(o)} + \sum_{i \in N \setminus (S \cup T)} \frac{u_i(o')}{u_i(o)} \right)\\
    &> \frac{1}{n} \left( \sum_{i \in S} \frac{|S \cup T|}{|S|} + \sum_{i \in N \setminus (S \cup T)} 1 \right) 
    \\
    &= \frac{1}{n} \left( |S \cup T| + |N \setminus (S \cup T)|\right)=1,
\end{align*}
contradicting that $o$ is PF.

$1$-\pgf $\Rightarrow$ $(\text{PO}+\text{Core} + 1\text{-\pef})$. Core, PO, and $1$-\pef all impose the same constraints as $1$-\pgf, but only for some of possible pairs of groups $S$ and $T$. PO imposes the constraint for $S=T=N$, the core imposes the constraints for $T=N$ (and any $S$), while $1$-\pef imposes them for $|S|=|T|=1$.
\end{proof}

Together with Proposition~\ref{prop:implications-prelim} and Theorem~\ref{thm:pef-ef}, Proposition~\ref{prop:implications-public} paints a clear picture of the relations among the various fairness notions in the public outcomes model, depicted in Figure~\ref{fig:hierarchy}b, mimicking a similar hierarchy for private goods division depicted in Figure~\ref{fig:hierarchy}a (except that $1$-\pef no longer implies proportionality in the public outcomes model). 

\begin{figure*}[htb!]
    \tikzset{
        mynode/.style={
            draw,
            rectangle,
            rounded corners=5pt,
            inner sep=5pt,
            minimum width=2.5cm
        },
        myarrow/.style={
            ->,
            -{Latex[length=2mm]}
        },
        mydoublearrow/.style={
            <->,
            -{Latex[length=2mm]}
        },
    }
    \centering
    \begin{subfigure}[t]{0.5\textwidth}
    \begin{tikzpicture}[scale=0.8, every node/.style={scale=0.9}]
    \node[mynode] (Y0) at (-0.5,1.5) {Proportional Fairness};
    \node[mynode] (Y00) at (4.5,1.5) {Max Nash Welfare};
    \node[mynode] (Y1) at (2,0) {Group Fairness};
    \node[mynode] (Y2) at (-0.5,-1.5) {Core};
    \node[mynode] (Y3) at (4,-1.5) {Group Envy-Freeness};
    \node[mynode] (Y4) at (-0.5,-3) {Proportionality};
    \node[mynode] (Y5) at (4,-3) {Envy-Freeness};
    \draw[myarrow] (Y0) -- node[near end, above] {*} ++ (Y1);
    \draw[myarrow] (Y00) -- node[near end, above] {*} ++ (Y1);
    \draw[myarrow] (Y0) -- (Y00);
    \draw[myarrow] (Y00) -- node[above] {*} ++ (Y0);
    \draw[myarrow] (Y1) -- (Y2);
    \draw[myarrow] (Y1) -- (Y3);
    \draw[myarrow] (Y2) -- (Y4);
    \draw[myarrow] (Y3) -- (Y5);
    \draw[myarrow] (Y5) -- (Y4);
    \end{tikzpicture}
    \caption{Private goods division.}
    \label{fig:hierarchy-private}
    \end{subfigure}%
    \begin{subfigure}[t]{0.5\textwidth}
    \begin{tikzpicture}[scale=0.8, every node/.style={scale=0.9}]
    \node[mynode] (Y0) at (-0.5,1.5) {Proportional Fairness};
    \node[mynode] (Y00) at (4.5,1.5) {Max Nash Welfare};
    \node[mynode] (Y1) at (2,0) {$1$-Group Harm Ratio};
    \node[mynode] (Y2) at (-0.5,-1.5) {Core};
    \node[mynode] (Y3) at (4.5,-1.5) {$1$-Equal-Sized Group Harm Ratio};
    \node[mynode] (Y4) at (-0.5,-3) {Proportionality};
    \node[mynode] (Y5) at (4.5,-3) {$1$-Individual Harm Ratio};
    \draw[myarrow] (Y0) -- (Y1);
    \draw[myarrow] (Y00) -- node[near end, above] {+} ++ (Y1);
    \draw[myarrow] (Y0) -- (Y00);
    \draw[myarrow] (Y00) -- node[above] {+} ++ (Y0);
    \draw[myarrow] (Y1) -- (Y2);
    \draw[myarrow] (Y1) -- (Y3);
    \draw[myarrow] (Y2) -- (Y4);
    \draw[myarrow] (Y3) -- (Y5);
    \end{tikzpicture}
    \caption{Public outcomes.}
    \label{fig:hierarchy-public}
    \end{subfigure}%
    \caption{The figure depicts a hierarchy of fairness notions for private goods division (already known) and the public outcomes model (based on our novel fairness definitions and results). For private goods division, the implications marked with (\textasteriskcentered{}) hold for cake-cutting~\cite{CFSW19,FSV20}. For the public outcomes model, we show the implications marked with (+) when the utility set $\U$ is compact and upper convex (Theorem~\ref{thm:mnw-pf}).}
    \label{fig:hierarchy}
\end{figure*}
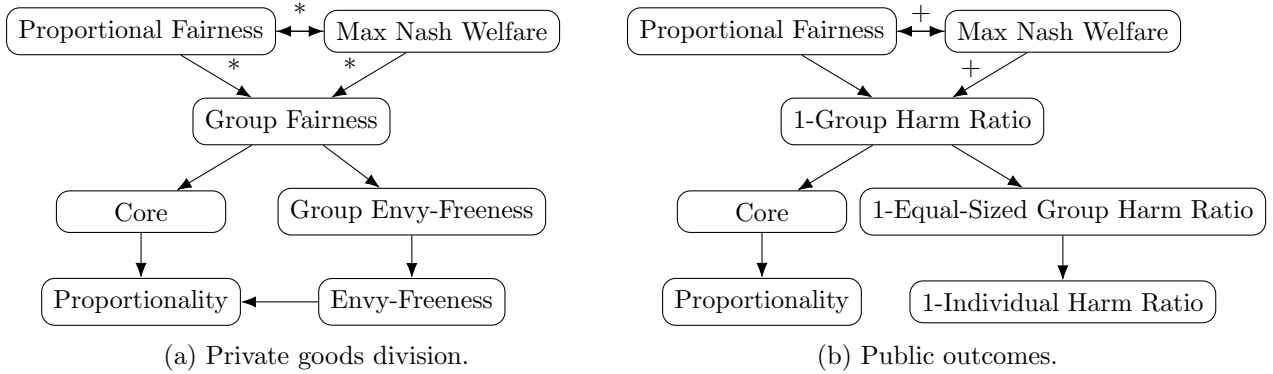

\section{Maximum Nash Welfare Solution}\label{sec:mnw}

Without any conditions on $\U$, an outcome satisfying any of the notions defined in \Cref{sec:prelim-fairness} is not guaranteed to exist. Proportionality cannot be guaranteed in the standard example of allocating an indivisible item between two agents who value it at $1$ ($\U=\{(1,0),(0,1)\}$). An MNW outcome is not guaranteed to exist if $x \in [0,1)$ fraction of a single homogeneous divisible good must be allocated to a single agent with utility $x$ for receiving $x$ fraction of the good ($\U = \{x : x \in [0,1)\}$); in particular, since the agent is not allowed to receive the whole good, a Nash welfare arbitrarily close to $1$ can be attained but a Nash welfare of $1$ cannot be attained. Nonexistence of the other properties follows from the implications in Proposition~\ref{prop:implications-prelim}.

In this section, we show that even proportional fairness, which implies a \pgffull of $1$, can be guaranteed when the utility set $\U$ satisfies two simple conditions:
\begin{itemize}
    \item \textbf{Compactness:} $\U$ is compact.
    \item \textbf{Upper convexity:} For any $\xx,\yy \in \U$ and $\alpha \in [0,1]$, there is $\zz \in \U$ such that $z_i \ge \alpha \cdot x_i + (1-
    \alpha) \cdot y_i$ for all $i \in N$.
\end{itemize}

Compactness and convexity (where even $\zz = \alpha \cdot \xx + (1-
\alpha) \cdot \yy \in \U$ is guaranteed) of the utility set is a common feature of many decision-making models. Compactness of $\U$ holds whenever the outcome set $\O$ is compact and the utility functions $u_i$ are continuous. Convexity of $\U$ holds whenever the outcome set $\O$ is the set of probability distributions (over some set) and one uses expected utilities~\cite{Young95}. In resource allocation contexts, it also holds for the richer class of concave utilities~\cite[Proposition 1]{BFT11}. The celebrated result of \citet{DS61} establishes compactness and convexity of the utility set for cake-cutting.

Upper convexity is a slightly relaxed version of convexity, which allows convexity to fail as long as there is a utility vector that Pareto dominates the one obtained from a convex combination. Intuitively, this relaxed requirement should not cause any issues because one should always be able to switch to such a Pareto dominating outcome, although a priori it is not clear if such an outcome would retain any fairness guarantees.\footnote{For example, in private goods division, it is well known that envy-freeness might not be preserved under Pareto improvements.} Upper convexity is a useful relaxation.

\begin{example}
Consider allocating a single homogeneous divisible good between two agents. The outcome set is the simplex $\O = \{(x_1,x_2) : x_1,x_2 \ge 0, x_1+x_2 = 1\}$. Suppose the agents have utility functions that are linearly increasing in the fraction of resource that they receive, but plateau when at least half of the resource is allocated to them. That is, $u_1((x_1,x_2)) = \min(0.5,x_1)$ and $u_2((x_1,x_2)) = \min(0.5,x_2)$. We have $(0.5,0) \in \U$ and $(0,0.5) \in \U$, but $(0.25,0.25) \notin \U$. Upper convexity is still satisfied due to the fact that $(0.5,0.5) \in \U$, which Pareto dominates $(0.25,0.25)$. 
\end{example}

Finally, it is worth noting that when $\U$ is compact, an MNW outcome is guaranteed to exist because the product is a continuous function, which attains a maximum over a compact set via Weierstrass' extreme value theorem. Similarly, when $\U$ is compact and upper convex, a proportional outcome is guaranteed to exist. For each agent $i$, let $\xx^i \in \argmax_{\xx \in \U} x_i$ be a utility vector with the highest utility for agent $i$ (compactness ensures that the maximum is attained in $\U$). Upper convexity guarantees the existence of a utility vector that (weakly) Pareto dominates $(1/n) \sum_{i \in N} \xx^i$, and this is proportional by definition. Note that because this outcome has a strictly positive Nash welfare, so must any MNW outcome (so it must give a positive utility to each agent).  

\begin{proposition}\label{prop:mnw-prop-exist}
    When $\U$ is compact, an MNW outcome exists. When $\U$ is compact and upper convex, a proportional outcome exists and every MNW outcome has a strictly positive Nash welfare.
\end{proposition}

However, the existence of a proportionally fair (PF) outcome or even an outcome in the core is not trivial to establish even under compactness and upper convexity of $\U$. The next result shows that under these conditions, every MNW outcome (which exists by Proposition~\ref{prop:mnw-prop-exist}) is PF, making MNW and PF equivalent due to Proposition~\ref{prop:implications-prelim} and yielding the existence of PF. Combining this with Proposition~\ref{prop:implications-public}, we get a $1$-\pgf (and thus a $1$-\pef) implication as well.

\begin{theorem}\label{thm:mnw-pf}
Consider any instance of the public outcomes model where the utility set $\U$ is compact and upper convex. A maximum Nash welfare (MNW) outcome exists, and every MNW outcome is proportionally fair (PF). Consequently, MNW and PF are equivalent, and every MNW outcome has a \pgffull (and thus an \peffull) of $1$. 
\end{theorem}

\begin{proof}
The existence of an MNW outcome is due to \Cref{prop:mnw-prop-exist}. Let $o \in \O$ be any MNW outcome. Suppose for contradiction that it is not PF. Then, there exists an outcome $o' \in \O$ such that 
\begin{align}\label{eqn:pf-violation}
    \frac{1}{n} \sum_{i \in N} \frac{u_i(o')}{u_i(o)} > 1.
\end{align}
Define a function $f : (0,1] \to \R$ such that for $\alpha \in (0,1]$, 
\[
f(\alpha) = \frac{1}{n} \cdot \sum_{i \in N} \log\big(\alpha \cdot u_i(o) + (1-\alpha) \cdot u_i(o')\big).
\]
Note that this is well-defined because $u_i(o) > 0$ for all agents $i \in N$ (\Cref{prop:mnw-prop-exist}) and $\alpha > 0$, implying that the first term inside the logarithm is strictly positive (and the second term is weakly positive since $\alpha \le 1$ and $u_i(o') \ge 0$). Note that $f(1) = \log \NW(o)$. 

Next, we prove that that there exists $\alpha \in (0,1)$ such that $f(\alpha) > f(1)$. Since $f$ is differentiable in $\alpha$, it suffices to show that the left-derivative of $f$ at $\alpha=1$ is strictly negative. It is easy to see that
\begin{align*}
    \partial_{-}f(1) &= \frac{1}{n} \sum_{i \in N} \frac{u_i(o) - u_i(o')}{u_i(o)} = 1-\frac{1}{n} \cdot \sum_{i \in N} \frac{u_i(o')}{u_i(o)} < 0,
\end{align*}
where the last inequality follows from \Cref{eqn:pf-violation}. 

Thus, there exists $\alpha \in (0,1)$ for which $f(\alpha) > f(1)$. Due to upper convexity of $\U$, there exists an outcome $\hat{o} \in \O$ such that $u_i(\hat{o}) \ge \alpha \cdot u_i(o) + (1-\alpha) \cdot u_i(o')$ for all agents $i \in N$. Thus, $\log \NW(\hat{o}) \ge f(\alpha) > f(1) = \log \NW(o)$, contradicting the fact that $o$ is an MNW outcome. Hence, we have proved that $o$ must also be PF. Since PF implies MNW more generally (\Cref{prop:implications-prelim}), this makes MNW and PF equivalent on every instance of the public outcomes model.
\end{proof}

For private goods division with additive utilities, $1$-\pgf implies PO and $1$-\pef (which in turn implies EF). Hence, Theorem~\ref{thm:mnw-pf} yields an alternative proof of the fact that every MNW allocation is EF and PO in cake-cutting, where the utility set is known to be compact and convex~\cite{DS61}. We remark that our elementary proof is much simpler than the proof by~\citet{SS19}, who proved it by first establishing that an MNW allocation is characterized by a form of market equilibrium (s-CEEI), and the proof by \citet{Wel85}, who showed only the existence of an EF+PO cake allocation using Kakutani's fixed point theorem.

The equivalence between PF and MNW is a well-known fact in many domains, going back to the work of \citet{Kelly97}, who proved it in the context of rate control in communication networks. The proof provided above is very similar, except we strive to not make any assumptions on the nature of the decision making (the outcome set $\O$) or the utility functions of the agents. The key to our proof is defining the function $f$ in terms of the scalar $\alpha$, which makes it differentiable without requiring any perturbations of utilities. 

Once MNW is established as equivalent to PF, one may wonder how much stronger MNW is compared to the next property in the fairness hierarchy, namely $1$-\pgf. For private goods division, the equivalent question would be how strong MNW is compared to group fairness (GF)~\cite{CFSW19}. \citet{FSV20} study allocation rules that map each instance to a set of (tied) allocations, and an allocation rule is said to satisfy GF if, on each instance, every allocation in the set returned by the rule satisfies GF. They prove that the rule returning the set of MNW allocations satisfies GF even for cake-cutting, and, subject to an axiom called replication invariance~\cite{DS63,Var74}, it is the only rule satisfying GF. 

Replication invariance informally requires that if one replicates an instance by creating $k$ copies of each agent and good, then a replication of every allocation returned by the rule on the original instance must be returned by the rule on the new instance. It is difficult to extend to the public outcomes model. 

Nonetheless, the next result shows that a \pgffull of $1$ implies a (tight) $2$-approximation of maximum Nash welfare in the public outcomes model. This result nicely complements that of \citet{FSV20} and has two key advantages: first, it eliminates the need to impose any additional axioms, which significantly changes the technical arguments required, and second, the proof works on an instance-by-instance basis (rather than at the level of rules) in the very general public outcomes model.

\begin{theorem}\label{thm:nw-apx}
Fix any instance of the public outcomes model. For every outcome $o \in \O$ with a \pgffull of $1$ (i.e., $1$-\pgf), we have that $\NW(o) \ge \binom{n}{\lfloor{n/2}\rfloor}^{-1/n} \cdot \NW(o')$ for all $o' \in \O$, and this is tight. Finally, $\binom{n}{\lfloor{n/2}\rfloor}^{1/n} \le 2$ and approaches $2$ as $n \to \infty$.
\end{theorem}

\begin{proof}
Fix any instance of the public outcomes model. Let $o,o' \in \O$ be any two outcomes satisfying $1$-\pgf. We want to show that $\NW(o) \ge \binom{n}{\lfloor{n/2}\rfloor}^{-1/n} \cdot \NW(o')$. If $\NW(o') = 0$, this holds trivially, so let us assume that $\NW(o') > 0$, i.e., $u_i(o') > 0$ for all agents $i \in N$. 

Sort the agents $i \in N$ in a non-decreasing order of $\frac{u_i(o)}{u_i(o')}$ (this is well-defined because $u_i(o') > 0$ for all $i \in N$); without loss of generality, let us rename the agents so that $\frac{u_1(o)}{u_1(o')} \le \frac{u_2(o)}{u_2(o')} \le \ldots \le \frac{u_n(o)}{u_n(o')}$. Let $L = \{i \in N : u_i(o') > u_i(o)\}$ and $H = \{i \in N : u_i(o') \le u_i(o)\}$. By our renaming, we have that $L = \{1,2,\ldots,k\}$ and $H = \{k+1,k+2,\ldots,n\}$, where $k=|L|$. 

Fix any $i \in L$. Let $S = \{1,2,\ldots,i\} \subseteq L$ and $T = H$. Note that $S \cap T = \emptyset$. Since $o$ is $1$-\pgf, there exists $j \in S$ for which $\frac{|S|}{|S \cup T|} \cdot u_j(o') \le u_j(o)$, i.e., $\frac{u_j(o)}{u_j(o')} \ge \frac{i}{n-k+i}$. Due to the sorting of the agents, this implies $\frac{u_i(o)}{u_i(o')} \ge \frac{i}{n-k+i}$. For any $i \in H$, we have $\frac{u_i(o)}{u_i(o')} \ge 1$ by the definition of $H$. Putting these together, we have
\begin{align*}
\frac{\NW(o)}{\NW(o')} 
&= \left(\prod_{i \in N} \frac{u_i(o)}{u_i(o')}\right)^{1/n}
\ge \left(\prod_{i=1}^k \frac{i}{n-k+i} \cdot \prod_{i=k+1}^n 1\right)^{1/n}
\\
&= \left( \frac{k!(n-k)!}{n!}\right)^{1/n} = \binom{n}{k}^{- {1/n}} \ge \binom{n}{\lfloor{n/2}\rfloor}^{-{1/n}},
\end{align*}
as required. 

To prove tightness, fix an $\epsilon \in (0,\nicefrac{1}{2})$ and consider an example with $\O = \{o,o^*\}$, where 
\[
u_i(o) = \begin{cases}
    \frac{i}{\lceil{n/2}\rceil+i} + \epsilon, &i \le \lfloor{n/2}\rfloor\\
    1+\epsilon, &i > \lfloor{n/2}\rfloor,
\end{cases}
\]
and $u_i(o^*) = 1$ for all $i \in N$. 
We claim that $o$ is $1$-\pgf. To see this, suppose for contradiction that there exist groups of agents $S,T$ for which the two conditions from \Cref{def:pgf} hold (taking $o'=o^*$). Note that $T \supseteq \{ i : i > \lfloor{n/2}\rfloor \}$, since any agent who receives less utility in outcome $o^*$ than in outcome $o$ must be contained in $T$. In particular, $|T| \ge \lceil{n/2}\rceil$. Next, observe that $S \subseteq \{ i : i \le \lfloor{n/2}\rfloor \}$, since agents in $S$ must receive more utility in outcome $o^*$ than in outcome $o$. In particular, writing $|S|=k$, we have $k \le \lfloor{n/2}\rfloor$ and $| S \cup T | \ge k + \lceil{n/2}\rceil$. Additionally, we have $\max_{i \in S} u_i(o) \ge \frac{k}{\lceil{n/2}\rceil+k}+\epsilon$. Letting $i' = \argmax_{i \in S} u_i(o)$,
\begin{align*}
    \frac{|S|}{|S \cup T|} \cdot u_i(o^*) \le \frac{k}{k + \lceil{n/2}\rceil} < \frac{k}{\lceil{n/2}\rceil+k}+\epsilon \le u_{i'}(o),
\end{align*}
contradicting the violation of $1$-\pgf.
To obtain the desired bound, note that 
\begin{align*}
    \frac{\NW(o)}{\NW(o^*)} \xrightarrow{\epsilon \to 0} &\left( \prod_{i=1}^{\lfloor{n/2}\rfloor} \frac{i}{\lceil{n/2}\rceil+i} \right)^{1/n}
    = \binom{n}{\lfloor{n/2}\rfloor}^{-1/n}.
\end{align*}
Finally, $\binom{n}{\lfloor{n/2}\rfloor}^{1/n} \le 2$ holds because there are only $2^n$ subsets (of any size) of a set of size $n$, and the limiting value as $n \to \infty$ follows directly from applying Stirling's approximation.
\end{proof}

When $\U$ is compact and upper convex, an MNW outcome exists (Proposition~\ref{prop:mnw-prop-exist}), and Theorem~\ref{thm:nw-apx} shows that any $1$-\pgf outcome achieves a $\binom{n}{\lfloor{n/2}\rfloor}^{1/n}$-approximation of its Nash welfare. 

\section{Computation}\label{sec:computation}
Having established that MNW (equivalently, PF) outcomes are $1$-\pgf, and thus $1$-\pef, under quite general conditions (Theorem~\ref{thm:mnw-pf}), we now examine the question of computation. We are interested in two computational questions: checking whether a given outcome is $1$-\pef and computing an outcome with a low \peffull (ideally, $1$). 

For the special case of homogeneous divisible goods with additive utilities, an MNW/PF solution can be found in strongly polynomial time~\cite{orlin2010improved}. However, in general, exact polynomial time computation is not possible. We instead turn to approximate solutions.

\begin{definition}[$\alpha$-Proportional Fairness]
    An outcome $o \in \outcomes$ is $\alpha$-Proportionally Fair ($\alpha$-PF) if $\frac{1}{n} \sum_{i \in N} \frac{u_i(o')}{u_i(o)} \le \alpha$ for all $o' \in \outcomes$.
\end{definition}

When the utility space $\U$ is convex and can be described by a polynomial number of linear constraints, as is the case for common settings such as public decision making~\cite{CFS17}, and budget-feasible participatory budgeting~\cite{FGM16} (in all cases, assuming fractional solutions are allowed), a $(1+\epsilon)$-PF outcome can be computed in time polynomial in the input size and $1/\epsilon$ using standard convex programming techniques.\footnote{To see this, note that the function $\text{PF}(\xx,\yy)=\frac{1}{n}\sum_{i \in N} \frac{y_i}{x_i}$ is convex in $\xx$. Therefore, since the supremum of convex functions is itself convex, $\text{PF}(\xx)=\max_{\yy \in \U} \frac{1}{n} \sum_{i \in N} \frac{y_i}{x_i}$ is convex in $\xx$. Thus, $\min_{\xx \in \U} \text{PF}(\xx)$ is a convex optimization problem with an optimal objective value of $1$. It follows that a solution with objective value of $1+\epsilon$ can be computed in time polynomial in the input size and $1/\epsilon$.} However, approximate PF is interesting only insofar as it guarantees some approximation of other, normative, properties. Our most general of these is \pgf. 

By slightly adapting the proof of Proposition~\ref{prop:implications-public}, it is easy to see that approximate PF implies approximate \pgf.

\begin{proposition}
\label{prop:apx-pf-pgf}
    In the public outcomes model, $(1+\epsilon)$-PF $\Rightarrow$ $(1+n\epsilon)$-\pgf for every $\epsilon \ge 0$.
\end{proposition}
\begin{proof}
Fix any $\epsilon \ge 0$. Suppose $o \in \O$ is $(1+\epsilon)$-PF but violates $(1+n\epsilon)$-\pgf. Then there exist non-empty $S, T \subseteq N$ and outcome $o' \in \O$ such that $\frac{|S|}{|S \cup T|} \cdot u_i(o') \ge (1+n\epsilon)u_i(o)$ for all $i \in S$, and $u_i(o') \ge u_i(o)$ for all $i \in N \setminus (S \cup T)$. If $u_i(o) = 0$ for any $i \in N$, then $o$ clearly violates $(1+\epsilon)$-PF, which is the desired contradiction. 

Assume that $u_i(o) > 0$ for all $i \in N$. Then we have
\begin{align*}
    \frac{1}{n} \sum_{i \in N} \frac{u_i(o')}{u_i(o)} &\ge \frac{1}{n} \left( \sum_{i \in S} \frac{u_i(o')}{u_i(o)} + \sum_{i \in N \setminus (S \cup T)} \frac{u_i(o')}{u_i(o)} \right)\\
    &\ge \frac{1}{n} \left( \sum_{i \in S} \frac{(1+n\epsilon)|S \cup T|}{|S|} + \sum_{i \in N \setminus (S \cup T)} 1 \right)\\
    &= \frac{1+n\epsilon}{n} \cdot |S \cup T| + \frac{1}{n} \cdot |N \setminus (S \cup T)|\\
    &\ge 1+\epsilon \cdot |S \cup T| \ge 1 + \epsilon,
\end{align*}
contradicting that $o$ is $(1+\epsilon)$-PF.    
\end{proof}

Of course, our results have not ruled out the possibility of computing an exactly $1$-\pef, or even a $1$-\pgf, outcome in polynomial time; perhaps this can be done more directly without having to appeal to (approximate) proportional fairness. We leave these as open questions. 

With respect to checking whether a given outcome $o \in O$ is $1$-\pef (i.e., has a \peffull of $1$), whenever the feasible utility space $\U$ can be described by a polynomial number of linear constraints, the problem can be solved in polynomial time. To do so, we solve a linear feasibility program for every pair of agents $i,j$ to check if there exists $\xx \in \U$ for which $x_i > 2 \cdot u_i(o)$ and $x_k \ge u_k(o)$ for all $k \in N \setminus \{ i,j \}$, which is a single linear program.

\textbf{Indivisible goods.} An interesting special case is that of indivisible goods allocation with additive utilities. Here, it is known that a $1$-\pef outcome cannot be guaranteed. For the case of two agents, we know that EF and $1$-\pef are equivalent (Theorem~\ref{thm:pef-ef}), and it is known via a standard reduction from the Partition problem that determining the existence of an EF outcome is NP-complete even with two agents.

\begin{corollary}
For indivisible goods allocation with additive utilities, checking if a $1$-\pef outcome exists is NP-complete. 
\end{corollary}

Lastly, we show that it is hard to check whether a given allocation of indivisible goods is $1$-\pef. This is in sharp contrast with checking envy-freeness, which can be done easily by comparing the utility of every agent with their value for every other agent's allocation.

\begin{table*}[t]
    \centering
    \renewcommand{\arraystretch}{1.25}
\begin{tabular}{|c|c|c c c|c c c|c c c|c c c|c c c|c|}  
  \hline  
  \multicolumn{2}{|c|}{} &\multicolumn{3}{c|}{$G^1$} & \multicolumn{3}{c|}{$G^2$} & \multicolumn{3}{c|}{$G^3$} & \multicolumn{3}{c|}{$G^4$} & \multicolumn{3}{c|}{$G^5$} &  \\  
  \cline{3-18}
  \multicolumn{2}{|c|}{} &$g_1^1$ & ... & $g_{3d}^1$ & $g_1^2$ & ... & $g_d^2$ & $g_1^3$ & ... & $g_d^3$ & $g_1^4$ & ... & $g_{3d}^4$ & $g_1^5$ & ... & $g_{3d}^5$ & $g^*$\\  
  \hline  
  \multirow{3}{*}{$Y$} & $y_1$ & \circled{1} & $\ldots$ & 0 & 0 & $\ldots$ & 0 & 0 & $\ldots$ & 0 & $1-\epsilon$ & & 0 & $\epsilon$ & & 0 & 0 \\
  & $\vdots$ & $\vdots$ & $\ddots$ & $\vdots$ & $\vdots$ & $\ddots$ & $\vdots$ & $\vdots$ & $\ddots$ & $\vdots$ & $\vdots$ & $\ddots$ & $\vdots$ & $\vdots$ & $\ddots$ & $\vdots$ & $\vdots$\\
  & $y_{3d}$ & 0 & $\ldots$ & \circled{1} & 0 & $\ldots$ & 0 & 0 & $\ldots$ & 0 & 0 & & $1-\epsilon$ & 0 & & $\epsilon$ & 0 \\
   \hline  
  \multirow{3}{*}{$Z$} & $z_1$ &$n_1$& $\ldots$ &$n_{3d}$& \circled{1} & $\ldots$ & 0 & $\circled{d}$& $\ldots$ & 0 & 0 & $\ldots$ & 0 & 0 & $\ldots$ & 0 & 0  \\
  & $\vdots$ & $\vdots$ & $\ddots$ & $\vdots$ & $\vdots$ & $\ddots$ & $\vdots$ & $\vdots$ & $\ddots$ & $\vdots$ & $\vdots$ & $\ddots$ & $\vdots$ & $\vdots$ & $\ddots$ & $\vdots$ & $\vdots$ \\
  & $z_d$ &$n_1$& $\ldots$&$n_{3d}$& 0 & $\ldots$ &$\circled{1}$& 0 & $\ldots$ &$\circled{d}$& 0 & $\ldots$ & 0 & 0 & $\ldots$ & 0 & 0 \\
  \hline
   \multirow{2}{*}{$W$} & $w_1$ & 0 & $\ldots$ & 0 &2& $\ldots$ &2& 0 & $\ldots$ & 0 &$\circled{\frac{1}{3}}$& $\ldots$ &$\circled{\frac{1}{3}}$& 0 & $\ldots$ & 0 & $\epsilon$ \\
  & $w_2$ & 0 & $\ldots$ & 0 &0 & $\ldots$ & 0 & 0 & $\ldots$ & 0 & 0 & $\ldots$ & 0 &$\circled{\epsilon}$ & $\ldots$ & $\circled{\epsilon}$& $\circled{\epsilon}$ \\
  \hline
\end{tabular} 
\caption{The indivisible goods allocation instance from the proof of Theorem~\ref{thm:check-pEF}. The outcome $A$ is indicated by circles.}
\label{tab:check-pEF}
\end{table*}

\begin{theorem}
\label{thm:check-pEF}
For indivisible goods allocation with additive utilities, checking if a given outcome is $1$-\pef is coNP-complete.
\end{theorem}
\begin{proof}
First, note that membership in coNP holds since an outcome $o$ can be certified to not be $1$-\pef by providing an alternative outcome $o'$ and two agents $i,j$ for which the violating conditions in \Cref{def:pub-ef} hold.

To show coNP-completeness, we reduce from the 3-Partition problem: Given a multiset $C$ of $3d$ numbers $c_1$, \ldots , $c_{3d}$ lying
strictly between $1/4$ and $1/2$, can $C$ be partitioned into $d$ triplets $C_1$, \ldots,  $C_d$ such that the sum of
the members of each triplet is 1?

Given a 3-Partition instance, we define an indivisible goods instance (along with an outcome) with $m=11d+1$ goods arranged into the following subsets: 
\renewcommand{\arraystretch}{1.2}
\[
\begin{array}{c@{\quad}c@{\quad}c}
G^1 = \{ g^1_1, \ldots, g^1_{3d}\},& G^2 = \{ g^2_1, \ldots, g^2_d\}, & G^3 = \{ g^3_1, \ldots, g^3_d\}, \\
G^4 = \{ g^4_1, \ldots, g^4_{3d}\}, &
G^5 = \{ g^5_1, \ldots, g^5_{3d}\},
 \\
\end{array}
\]
along with a single good $g^*$. There are $n = 4d+2$ agents, arranged into the following subsets:
\begin{align*}
Y 
= \{ y_1, \ldots, y_{3d}\}, \quad
Z 
= \{ z_1, \ldots, z_d \}, \quad 
W 
= \{ w_1, w_2 \}.
\end{align*}
Let $0<\epsilon<1$. Valuations that the agents have for the goods are defined as follows (for ease of reading, we use the notation $v_i(g)$ rather than $v_{i,g}$; valuations that are not directly specified are equal to zero). For every $i \in Y$, $v_i(g^1_i)=1$, $v_i(g^4_i)=1-\epsilon$, and $v_i(g^5_i)=\epsilon$. For every $i \in Z$ and every $j \in [3d]$, $v_i(g^1_j)=n_j$. Thus, these agents and items directly encode the values from the 3-Partition instance. Further, for every $i \in Z$, $v_i(g^2_i)=1$, and $v_i(g^3_i)=d$. For agent $w_1$ we have $u_{w_1}(g^2_j)=2$ for all $j \in [d]$, $u_{w_1}(g^4_j)=1/3$ for all $j \in [3d]$, and $u_{w_1}(g^*)=\epsilon$. For agent $w_2$, we have $u_{w_2}(g^5_j)=\epsilon$ for all $j \in [3d]$ and $u_{w_2}(g^*)=\epsilon$. 

The allocation $A$ is as follows, and is also depicted in \Cref{tab:check-pEF}. For every $j \in [3d]$, good $g^1_j$ is allocated to agent $y_j$. For every $j \in [d]$, goods $g^2_j$ and $g^3_j$ are allocated to agent $z_j$. Agent $w_1$ is allocated the entire set of goods $G^4$. Finally, agent $w_2$ is allocated the entire set of goods $G^5$ as well as good $g^*$. We have $u_{y_i}(A_{y_i})=1$ for all $i \in [3d]$, $u_{z_i}(A_{z_i})=d+1$ for all $i \in [d]$, $u_{w_1}(A_{w_1})=d$, and $u_{w_2}(A_{w_2})=(3d+1)\epsilon$. Note that for all agents $i \neq w_1$, it is the case that $u_i(A_i) \ge u_i(M)/2$.

First suppose that there exists a solution to the 3-Partition instance. Then we can define an outcome $A'$ as follows. For every $i \in [3d]$, agent $y_i \in Y$ is allocated $g^4_i$ and $g^5_i$. For every $i \in [m]$, agent $z_i \in Z$ is allocated $g^3_i$ and a set of three goods from $G^1$ with total value exactly equal to 1 (this is possible by the assumption of a 3-Partition solution). Agent $w_1$ is allocated $g^*$ and all goods in $G^2$. We have $u_{y_i}(A'_{y_i})=1=u_{y_i}(A_{y_i})$ for all $i \in [3d]$, $u_{z_i}(A'_{z_i})=d+1=u_{z_i}(A_{z_i})$ for all $i \in [m]$, and $u_{w_1}(A'_{w_1})=2m+\epsilon > 2d=2u_{w_1}(A_{w_1})$. Since only agent $w_2$ receives less utility under $A'$ than under $A$, allocation $A'$ is witness to a violation of $1$-\pef.

Next, suppose that there does not exist a solution to the 3-Partition instance. We will show that outcome $A$ is $1$-\pef.

Since $u_i(A_i) \ge u_i(M)/2$ for all $i \neq w_1$, the only agent who can more than double her utility in some outcome $A' \neq A$ is $w_1$. Consider such an outcome $A'$ and suppose for contradiction that $A'$ witnesses a violation of $1$-\pef. 

For it to be the case that $u_{w_1}(A'_{w_1})>2u_{w_1}(A_{w_1})=2m$, agent $w_1$ must be allocated some goods from $G^2$ (each of which she has utility 2 for). Suppose that she receives $k$ such goods. We may assume that $m \ge 3$ and therefore receiving only $k=1$ good from $G^2$ is not sufficient to double her utility; assume therefore that $k \ge 2$. We now follow a chain of implications about what this means for the outcome $A'$, given that there can be at most one agent $i$ with $u_i(A'_i) < u_i(A_i)$.

At least $k$ agents in $Z$ must not be allocated any $g^2_j \in G^2$ under allocation $A'$. Denote this set of agents by $Z'$. For $A'$ to witness a violation of $1$-\pef, it must be the case that for at least $k-1 \ge 1$ of these agents in $Z'$, they receive utility at least $u_{z_i}(A'_{z_i}) \ge d+1$, which implies that $A'_{z_i}= \{ g^3_{z_i} \} \cup X$ for some $X \subseteq G^1$ with $|X| \ge 3$. The condition $X \subseteq G^1$ is required because, other than $g^3_{z_i}$ (for which agent $z_i$ has utility $d$) and $g^2_{z_i}$ (which is not allocated to agent $z_i$, by definition), agents in $Z'$ only get positive utility from $G^1$, and $u_{z_i}(g^1_j)<1/2$ for all $j \in [3d]$. 

Now, there exist at least three agents $y_i \in Y$ for whom $g^1_i \not \in A'_{y_i}$. Since $A'$ witnesses a violation of $1$-\pef, it is necessary for at least two of these agents to have $u_{y_i}(A'_{y_i}) \ge u_{y_i}(A_{y_i}) =1$, which implies that $A'_{y_1} = \{ g^4_i, g^5_i \}$. In particular, if good $g^5_i$ is allocated to agent $y_i$ under $A'$, then it is no longer allocated to agent $w_2$, which means that $u_{w_2}(A'_{w_2}) < u_{w_2}(A_{w_2})$. Thus, for all agents $i \neq w_2$, it must be the case that $u_i(A'_i) \ge u_i(A_i)$, or else $A'$ would not witness a violation of $1$-\pef.

Having established this fact, the remaining analysis is simplified. Agent $w_1$ receives $k$ goods from $G^2$ under $A'$. Therefore, at least $3k$ goods from $G^1$ must be allocated to agents in $Z'$, to compensate each of them for their lost good from $G^2$. Accordingly, at least $3k$ goods from $G^4$ and $3k$ goods from $G^5$ must be allocated to agents in $Y$, to compensate those agents who lost a $G^2$ good. Since agent $w_1$ values each $G^4$ good at 1/3, her utility from $A'$ is $u_{w_1}(A'_{w_1}) \le 2k+(3d-3k)/3+\epsilon = d+k+\epsilon$ (the additional $\epsilon$ comes from allocating $g^*$ to $w_1$). Thus, unless $k=d$, $u_{w_1}(A'_{w_1})<2d$, contradicting the assumption that $A'$ is a witness to a $1$-\pef violation. So assume $k=d$. But in this case, all of the goods in $G^1$ need to be allocated to agents in $Z$. Since there does not exist a solution to the 3-Partition problem, it is impossible to assign them so that $u_{z_i}(A'_{z_i}) \ge d+1=u_{z_i}(A_{z_i})$ for all $i \in [d]$, since doing so would require each agent $z_i$ receiving utility exactly 1 from a subset of the $G^1$ goods. Again, this contradicts $A'$ witnessing a $1$-\pef violation. Therefore, $A$ is $1$-\pef.
\end{proof}

Our proof uses a variable number of agents. For two agents, $1$-\pef is equivalent to EF, so checking whether a given outcome is $1$-\pef is in P. Could it perhaps be in P for any constant number of agents? We leave this as an open question. 

\section{Experiments}\label{sec:experiments}
In this section, we validate our theoretical definitions by running numerical simulations on data from three real-world settings: private goods allocation, peer review, and participatory budgeting. 

Consistent with real-world constraints in these settings, we require that solutions be discrete. In this case, the utility sets are not upper convex and it is therefore no longer the case that an outcome with an \pef/\pgf of at most 1---or even a finite \pef/\pgf---always exists.
Our approach will be to evaluate the individual and group harm ratios of prominent rules for these settings.

As a minimal requirement, our definition should differentiate between rules across a wide range of settings, since the practical consequence of any fairness definition is limited to those cases where it can make a specific recommendation. Second, in settings where well-established and validated fairness notions already exist, we would like \pef to align well with those existing notions.\footnote{Of course, the strength of \pef lies in its broad applicability to all these settings.} 

\begin{figure*}[htb!]
    \centering
    \begin{subfigure}{0.34\textwidth}
        \centering
        \includegraphics[width=\textwidth]{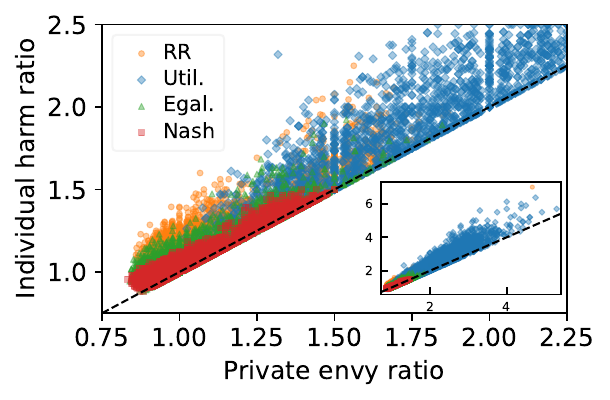}
        \caption{Uniform Multinomial}
        \label{fig:random-scatter}
    \end{subfigure}
    \begin{subfigure}{0.31\textwidth}
        \centering
        \includegraphics[width=\textwidth]{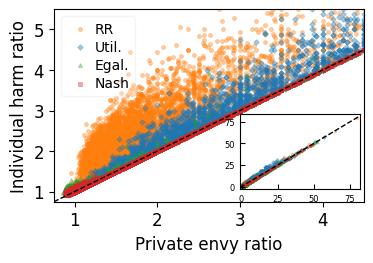}
        \caption{Dirichlet-Multinomial}
        \label{fig:market-scatter}
    \end{subfigure}
    \begin{subfigure}{0.31\textwidth}
        \centering
        \includegraphics[width=\textwidth]{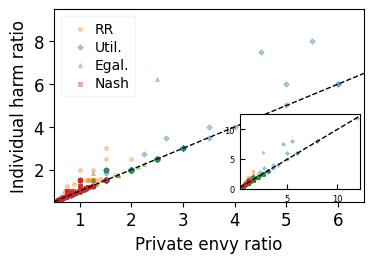}
        \caption{Spliddit}
        \label{fig:spliddit-scatter}
    \end{subfigure}
    \\
    \begin{subfigure}{0.325\textwidth}
        \centering
        \includegraphics[width=\textwidth]{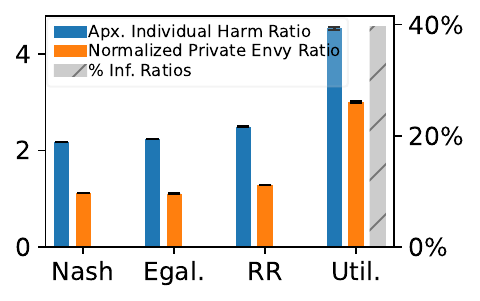}
        \caption{Uniform Multinomial}
        \label{fig:random-bar}
    \end{subfigure}
    \begin{subfigure}{0.325\textwidth}
        \centering
        \includegraphics[width=\textwidth]{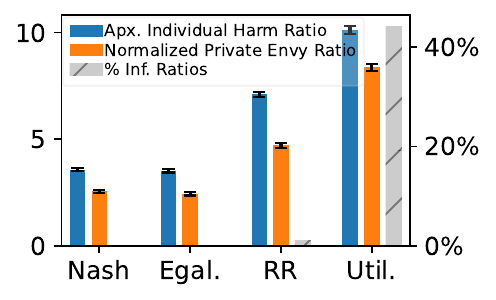}
        \caption{Dirichlet-Multinomial}
        \label{fig:market-bar}
    \end{subfigure}
    \begin{subfigure}{0.325\textwidth}
        \centering
        \includegraphics[width=\textwidth]{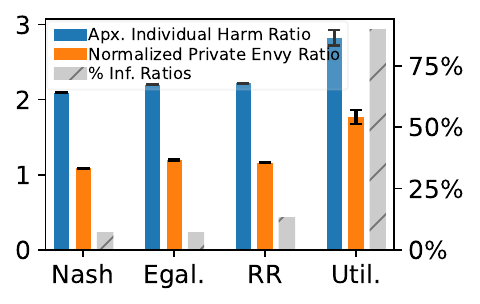}
        \caption{Spliddit}
        \label{fig:spliddit-bar}
    \end{subfigure}
    \caption{Scatter plots showing instance-wise \peffull and private envy ratio, and bar charts displaying their averages as well as the fraction of instances with infinite ratios. In the scatter plots, the inset plot shows the full set of instances, while the main plot is zoomed into its lower left region that contains approximately 90\% of the instances.}
    \label{fig:private-goods-plots}
\end{figure*}

\subsection{Private Goods}
Our first setting is that of indivisible private goods allocation. Every agent $i$ has a value $v_{i,g}$ for every good $g$, with valuations normalized so that $\sum_g v_{i,g}$ is independent of $i$. 

\textbf{Datasets.} We use two synthetic and one real datasets. For the synthetic ones, we first consider agent valuations drawn from a uniform multinomial distribution (i.e., randomly dividing $200$ points between $m$ goods) and, second, valuations drawn from a Dirichlet-multinomial distribution. The Dirichlet-multinomial can be thought of as first drawing an underlying ``market value'' for every good from a uniform Dirichlet distribution, and then drawing agent valuations from a multinomial distribution whose expectations equal to the market values, thus correlating agent valuations. For the synthetic experiments, we vary $n \in [3, 16]$ and $m \in [n, 5n]$ and generate 10 samples for each combination of $n$ and  $m$. Finally, we also use real data from the popular fair division website \href{www.spliddit.org}{Spliddit.org}~\cite{GP14}. 
\newcommand{\egalp}{Egal.\xspace}
\newcommand{\nashp}{Nash\xspace}
\newcommand{\utilp}{Util.\xspace}
\newcommand{\rrp}{RR\xspace}

\textbf{Rules.} For every fair division instance, we compute the outcomes that maximize the egalitarian welfare (minimum utility), the Nash welfare (product of utilities), and the utilitarian welfare (sum of utilities) --- we refer to these rules by \egalp, \nashp, \utilp respectively. We additionally consider round robin (\rrp), which fixes an ordering of the agents and allows each to pick one good at a time in a cyclic fashion; this, like Nash, is guaranteed to achieve an approximate envy-freeness guarantee~\cite{CKMP+19}.

\textbf{Measurements.} For each dataset, we display two graphs (\Cref{fig:private-goods-plots}). The first (Figure~\ref{fig:private-goods-plots}a to \ref{fig:private-goods-plots}c) is a scatter plot with each point being an outcome, color coded by the rule that produced it. The vertical axis measures the \peffull (\pef), and the horizontal axis measures the private envy ratio, defined as $\text{per} \triangleq \frac{1}{2} (1+\max_{i,j \in [n]} u_i(A_j)/u_i(A_i))$, which is the maximum multiplicative increase in utility that any agent could get by swapping their assigned bundle of goods with some other agent, linearly scaled so that an \pef of $\alpha$ guarantees $\text{per} \le \alpha$. Hence, all points are guaranteed to lie on or above the 45 degree line $y=x$. Note that EF guarantees a private envy ratio of at most $1$. 

The second graph for each dataset is a bar graph that shows, for each rule, the average IHR, the average private envy ratio, and the fraction of instances with infinite IHR and private envy ratio. Infinite ratios occur whenever some agent receives zero utility from their allocation, and the average ratios are taken over instances with finite ratios.

\textbf{Observations.} In all datasets, we observe that \peffull and private envy ratio are strongly correlated, with the points in the scatter plot lying somewhat close to the 45 degree line. The Nash and egalitarian rules give the lowest public and private envy ratios. Round robin is competitive on the uniform Dirichlet and Spliddit datasets, but less so for the Dirichlet-multinomial model. The utilitarian rule has high \peffull and private envy ratio, and a large fraction of instances with infinite ratios.

\subsection{Peer Review}
\textbf{Datasets.} In the peer review setting, we use real data from three conferences: ICLR 2018, CVPR 2017, and CVPR 2018. Each submission is an ``agent'' who has a value for being assigned a reviewer equal to a (system-generated) similarity score reflecting the reviewer's expertise in the paper's topic. Following the methodology of \citet{AMS23}, we assume that each submission has exactly one author, and that the set of reviewers coincides with the set of authors. The matching of reviewers to submissions is constrained in that no author can review their own submission, each submission must be assigned three distinct reviewers, and each reviewer must be assigned three submissions to review. 

We also follow the methodology of~\citet{AMS23} to infer authorship information. The ICLR 2018 dataset provides a conflict matrix, indicating when reviewer $j$ has a conflict with submission $i$. We find a maximum matching on the conflict matrix, and use this as the authorship matrix. We then subsample 300 instances of size 50 (that is, 50 submissions along with their 50 authors as the pool of reviewers). For the CVPR datasets, no conflict matrices are available. Instead, we again sample 300 instances of size 50 by randomly selecting one paper at a time. As its author, we select the author with the highest similarity score from those not already added to the pool of reviewers, and add that author to the pool of reviewers.

\textbf{Rules.} We compare five rules for peer review. The Toronto Paper Matching System (TPMS)~\citep{CZ13} maximizes the utilitarian welfare, and we also include the Nash welfare-maximizing assignment and the PeerReview4All (PR4A)~\citep{SSS21} algorithm, which is a heuristic algorithm for the leximin solution, which maximizes the minimum utility (egalitarian welfare) and, subject to that, maximizes the second-minimum, and so on. We also compare the Core-Based Reviewer Assignment (CoBRA) algorithm of~\citet{AMS23}, which computes a reviewer assignment that lies in the core, and an appropriate adaptation of the round robin rule. In the constrained setting of peer review, round robin works by letting agents choose a single reviewer in a cyclic fashion according to some fixed order, subject to their choice of reviewer leaving at least one feasible reviewer assignment remaining that is consistent with earlier choices.\footnote{We also assume that the minimum utility of any author for any reviewer is nonzero. Specifically, we set this minimum score to $10^{-3}$ while the maximum utility of each author for a reviewer is exactly $1$. This avoids infinite ratios, and it is reasonable to assume receiving more reviews (with TPMS score of 0) is at least slightly better than receiving fewer.}

\textbf{Measurements.} For the peer review setting the standard definition of envy does not make sense, since the reviewers allocated to paper $j$ may not be allowed to be allocated to paper $i$ (in particular, if the author of paper $i$ is a reviewer for paper $j$). Denoting $j$'s ``bundle'' of assigned reviewers by $A_j$, it might be the case that $u_i(A_j)$ is undefined. To avoid this issue, we consider a modified notion of envy ratio that we term the \emph{shuffle envy ratio}.\footnote{Similar results are obtained if we instead use the standard definition of envy, and only consider an agent's envy towards papers that they themselves are not reviewing.} For a pair of agents $i,j$, agent $i$'s shuffle envy ratio towards $j$ is half\footnote{This factor is chosen to align it with the definition of \pef, which also measures half of the ratio between the maximum utility of agent $i$ in any outcome that hurts no agent but agent $j$ and the current utility of agent $i$.} of the ratio between the maximum utility agent $i$ could receive from a feasible allocation $B$ such that $A_k=B_k$ for all $k \neq i,j$, and her utility for her assigned set of reviewers, $u_i(A_i)$. That is, agent $i$ is allowed to improve the outcome (from her perspective) by shuffling her own allocation and agent $j$'s allocation in whatever way she wishes, but cannot touch the allocation of any other agent. The shuffle envy ratio of a given allocation is the maximum shuffle envy ratio over all pairs of agents. Note that \pef must be at least as large as the shuffle envy ratio, since the modified allocations that shuffle envy allows are also allowed by the definition of \pef.

\begin{figure*}
    \centering
    \begin{subfigure}[b]{0.325\textwidth}
        \centering
        \includegraphics[width=\textwidth]{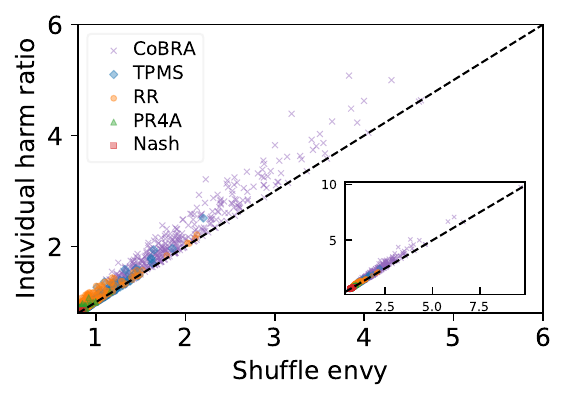}
        \caption{ICLR 2018}
    \end{subfigure}
    \begin{subfigure}[b]{0.325\textwidth}
        \centering
        \includegraphics[width=\textwidth]{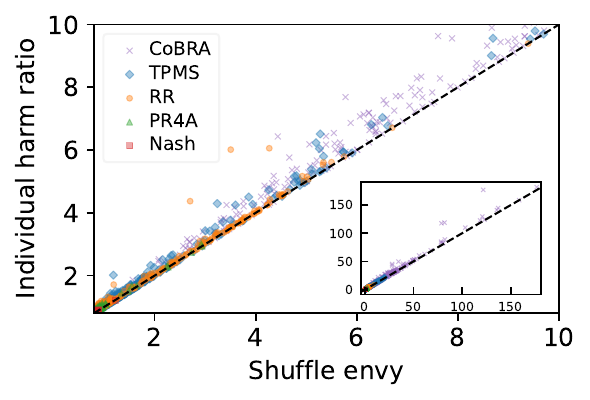}
        \caption{CVPR 2018}
    \end{subfigure}
    \begin{subfigure}[b]{0.325\textwidth}
        \centering
        \includegraphics[width=\textwidth]{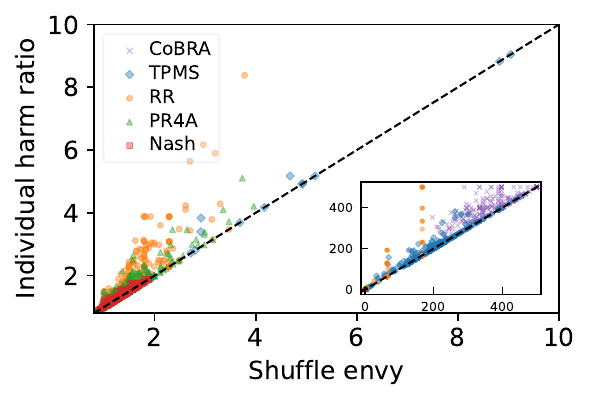}
        \caption{CVPR 2017}
    \end{subfigure}
    \begin{subfigure}[b]{0.3\textwidth}
        \centering
        \includegraphics[width=\textwidth]{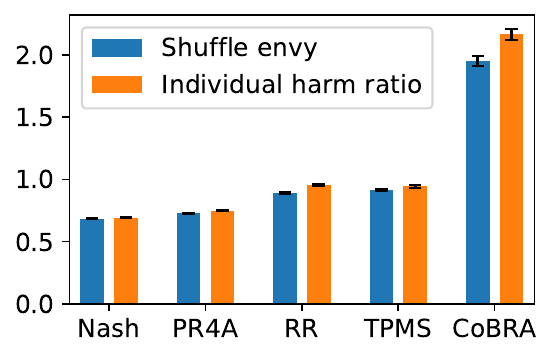}
        \caption{ICLR 2018}
    \end{subfigure}
    \quad
    \begin{subfigure}[b]{0.3\textwidth}
        \centering
        \includegraphics[width=\textwidth]{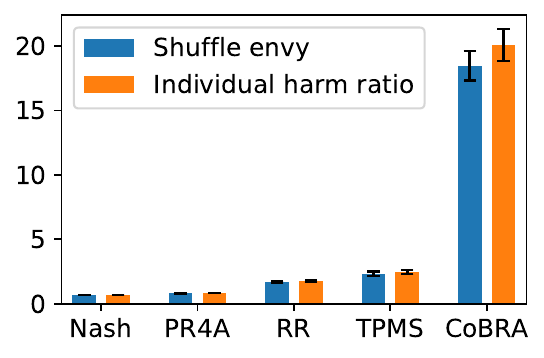}
        \caption{CVPR 2018}
    \end{subfigure}
    \quad
    \begin{subfigure}[b]{0.3\textwidth}
        \centering
        \includegraphics[width=\textwidth]{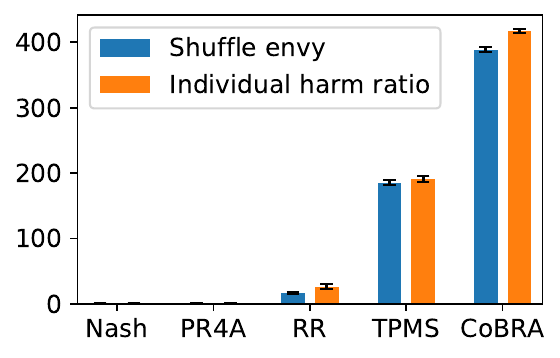}
        \caption{CVPR 2017}
    \end{subfigure}
    \caption{Scatter plots showing instance-wise public and shuffle envy ratios and bar charts displaying average public and shuffle envy ratios. In the scatter plots, the main plot is zoomed into the lower left region of the graph that contains more than 90\% of all data points -- except for CVPR'17 where CoBRA and TPMS lie far from the others -- while the inset plot shows the full set of instances.}
\end{figure*}

\textbf{Observations.} Examining the graphs, we again see that \pef is highly correlated with the shuffle envy ratio. Taken together with the private goods division results in the previous subsection, these results give us confidence that \pef is capturing an intuitive and meaningful aspect of fairness. In terms of the algorithms that we examine, Nash, PR4A, and round robin all perform well. TPMS does well on ICLR 2018 and CVPR 2018 data, but relatively badly on CVPR 2017. 
Finally, CoBRA performs much worse than all other algorithms in terms of both \pef and shuffle envy ratio.  

\subsection{Participatory Budgeting}
Our final setting is participatory budgeting, where the goal is for a city to select a subset of proposed projects to fund subject to a budget constraint, based on residents' preferences. Specifically, there is a set of projects, each with an associated cost. Each voter casts an approval ballot, approving a subset of the projects. The outcome is a set of projects to fund whose total cost is at most a given budget.  

\textbf{Datasets.} We use real participatory budgeting elections from \href{https://pabulib.org/}{Pabulib.org}~\cite{SST20}. We focus on five Polish cities: Warsaw, Gdynia, Lodz, Wroclaw, and Zabrze. For tractability, we use the 321 elections with less than 15 projects on the ballot. We consider two common models of voter utility: approval utilities, in which the utility that a voter has for an outcome is the number of her approved projects that are funded, and cost utilities, in which her utility is the total cost of her approved projects that are funded. We define a voter's utility for having none of her approved projects funded to be $\epsilon$ instead of $0$.\footnote{We set $\epsilon=10^{-2}$ for approval utilities, and $\epsilon=10^{-3} B$ for cost utilities, where $B$ is the total budget.} This avoids infinite harm ratios and arguably better reflects reality, where we would expect voters to derive at least a small amount of utility even for projects that they did not vote for.

\textbf{Measurements.} When the number of agents is large, as in citywide elections, the \peffull tends to be very small because it is hard for to find an alternative outcome that makes one agent better off, hurts at most one other agent, and keeps all the remaining (many) agents at least as happy. For this reason, we instead evaluate the \pgffull in this section, which is always at least as large as the \peffull. However, note that it is easy for a small group of agents to be harmed by the presence of a larger group. For example, a single agent who has no approved projects funded gets utility $\epsilon$, but can increase her utility by a very large multiplicative factor if she is allowed to hurt all other agents and simply choose her most preferred budget-feasible set of projects to fund. 

Recognizing this, we examine the average \pgf as a function of the size of the harmed group $S$. By the above argument, small groups $S$ will often produce large \pgf, but these are violations that an election organizer might be comfortable with in a large election; unfairness towards a large group of agents, on the other hand, remains undesirable. Technically, we measure \pgf as in Definition~\ref{def:pgf}, but only taking maximum over harmed groups $S$ of a fixed size. 

In \Cref{fig:pb}, we plot the average (over elections) of the maximum \pgffull across all $(S,T)$ subject to the condition that $|S|/n \ge x$. For comparison, in \Cref{fig:pb-percentiles}, we consider the full distribution of agent utilities, normalized by their maximum possible utility.\footnote{The ratio of an agent's utility to her maximum possible utility is the approximation of proportionality for her.} For every instance, we order the agents by this normalized utility, and take the average over instances (we give voter percentile on the $x$ axis to control for instances having different numbers of agents).

\begin{figure*}[t]
    \centering
    \begin{subfigure}{0.32\textwidth}
        \centering
       \includegraphics[width=\textwidth]{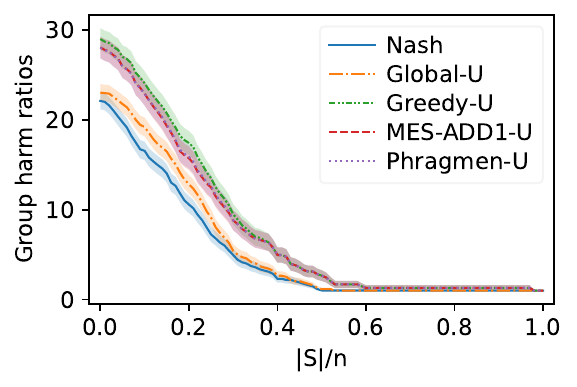}
        \caption{Approval Utilities}
    \end{subfigure}
    \qquad\qquad\qquad
    \begin{subfigure}{0.32\textwidth}
        \centering
        \includegraphics[width=\textwidth]{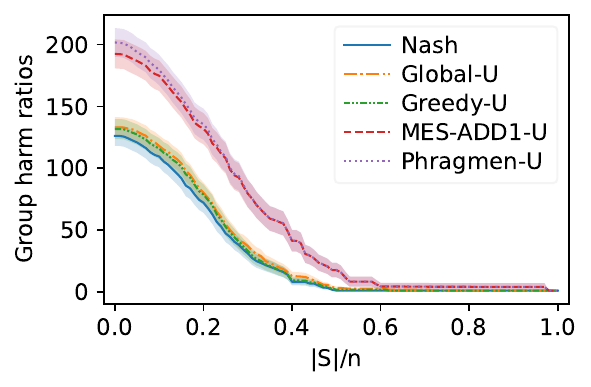}
        \caption{Cost Utilities}
    \end{subfigure}
    \caption{Plots show normalized public group envy ratio as a function of the minimum size of the envious group of agents $S$.}
    \label{fig:pb}
\end{figure*}
\begin{figure*}[t]
    \centering
    \begin{subfigure}{0.32\textwidth}
        \centering
        \includegraphics[width=\textwidth]{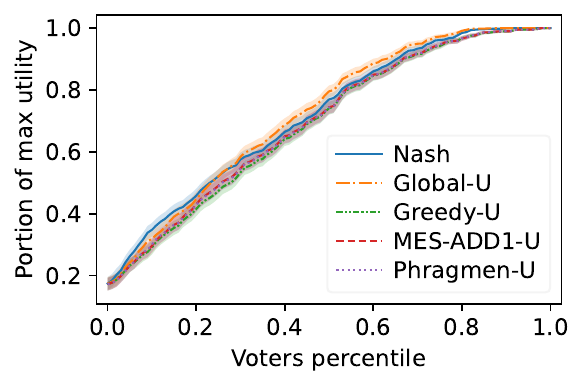}
        \caption{Approval Utilities}
    \end{subfigure}
    \qquad\qquad\qquad
    \begin{subfigure}{0.32\textwidth}
        \centering
        \includegraphics[width=\textwidth]{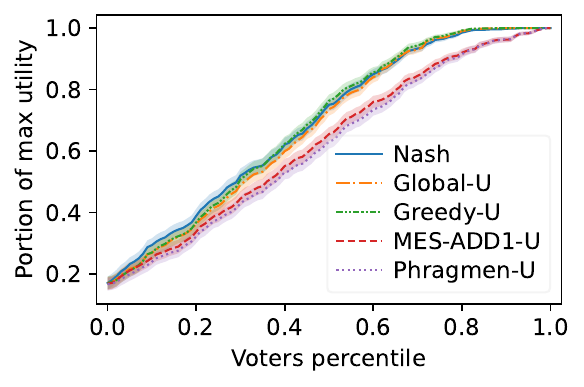}
        \caption{Cost Utilities}
    \end{subfigure}
    \caption{Plots show percentiles of voters proportional utilities (fraction of one's maximum feasible utility) averaged over all PB instances.}
    \label{fig:pb-percentiles}
\end{figure*}

\textbf{Rules.} We consider five commonly studied rules. Global-U is the utilitarian welfare-maximizing solution, Greedy-U is its greedy version that is more commonly used in practice~\cite{AS21}, and Nash denotes the MNW rule.\footnote{Greedy-U selects projects one by one in decreasing order of their utilitarian welfare, skipping over any project that cannot be added due to the budget constraint. In practice, instead of MNW, its two close approximations, namely proportional approval voting (PAV) and maximum smoothed Nash welfare, are more often used due to their better axiomatic properties. In \Cref{app:more-experiments}, we show that these rules have performance similar to that of MNW.} We also consider the method of equal shares (MES)~\cite{PPS21} and Phragmen's sequential rule (see, e.g., \cite{brill2023phragmen}). Importantly, since MES and Phragmen might return an outcome that does not exhaust the allowed budget, we complete the outcomes of these rules into maximal solutions using the greedy utilitarian method (which adds budget-feasible unselected projects one by one, in the decreasing order of their utilitarian welfare). MES is considered a very attractive rule for participatory budgeting and is known to satisfy compelling proportionality properties, while the utilitarian rules are typically considered to be less fair than other rules due to their focus on satisfying a majority of voters, potentially at the expense of a minority. 

\textbf{Observations.} In Figure~\ref{fig:pb}, we see that the \pgffull is high when small harmed groups $S$ are included, but decreases steadily as only larger harmed groups are considered. Under all rules and both utility models, only a very small \pgffull remains for harmed groups of size larger than $0.5n$. Perhaps most interesting is the relative performance of the algorithms. For both utility models, Nash yields the smallest harm ratios, closely followed by Global-U, the utilitarian solution. Somewhat surprisingly, MES and Phragmen, the two rules that are considered more fair in the literature, give outcomes with significantly higher \pgf. The reason is clear when considering \Cref{fig:pb-percentiles}, which shows a clear dominance of the utility vectors achieved by Nash and Global-U over those achieved by MES and Phragmen, and it is this dominance that is being captured by the \pgffull values in \Cref{fig:pb}. 

\section{Discussion}\label{sec:discussion}
Our proposal of novel fairness criteria opens the doors for a variety of extensions that have been popular in the fair division literature lately. For example, in integral settings where the utility set is discrete, can we guarantee ``up to one'' relaxations of our fairness criteria? Can the results be extended to weighted agents with unequal entitlements? Theorem~\ref{thm:mnw-pf} goes through with the usual weighted extensions of MNW ($\argmax_{o} \prod_{i \in \agents} u_i(o)^{w_i}$) and PF ($\max_{o'} \sum_{i \in \agents} w_i \cdot \frac{u_i(o')}{u_i(o)} \le 1$), along with a natural weighted extension of \pgffull, where the factor $\frac{|S|}{|S \cup T|}$ is replaced by $\frac{w(S)}{w(S \cup T)}$ (with $w(C) \triangleq \sum_{i \in C} w_i$). However, extending Theorem~\ref{thm:nw-apx} or considering the discrete case can be interesting. Finally, what can we say about the public bads model, where each agent has a non-positive utility for every outcome? 

\bibliographystyle{abbrvnat}
\bibliography{abb,ultimate,ultimate2,nisarg}

\newpage
\appendix
\section*{\Large \centering Appendix}

\section{Detailed Significance and Implications}\label{app:significance}
\paragraph{Simpler proof of envy-free and Pareto optimal cake-cutting.} We prove that maximum Nash welfare (MNW) achieves \pgffull (\pgf) under compactness and upper convexity of the feasible utility set $\U$ (\Cref{thm:mnw-pf}). As we elaborate in \Cref{sec:mnw}, for the prototypical fair division model of cake-cutting (formally defined in \Cref{sec:prelim-private}), these conditions hold~\cite{DS61} and \pgf implies envy-freeness (EF) and Pareto optimality (PO). Thus, our result provides an alternative proof of the celebrated result that MNW achieves EF+PO in cake-cutting~\cite{Wel85,SS19}. 

Beyond using the result of \citet{DS61} that the feasible utility set $\U$ is closed and convex for cake-cutting, our proof is fully elementary, e.g., teachable in a graduate course on the subject. Its domain-agnostic nature alleviates the need for defining domain-specific terminology (e.g., Borel sets and $\sigma$-algebra for cake-cutting), once closedness and convexity of $\U$ (a subset of $\R^n$) has been established. Finally, it is much more straightforward than the proof of existence of an EF+PO allocation by \citet{Wel85}, which uses Kakutani's fixed point theorem, and the proof of MNW satisfying EF+PO by \citet{SS19}, who first establish an equivalence between MNW and a market equilibrium concept called s-CEEI. 

\paragraph{MNW allocations are $2$-EF and PO when swaps are allowed.} \citet{hylland1979efficient} show existence of EF and PO outcomes in the one-sided matching setting, where $n$ agents need to be (fractionally) matched to $n$ items. This is equivalent to allocating $n$ private goods to $n$ agents but with the additional constraints that each agent receives a total fraction of $1$ from all the goods. Very recently, \citet{troebst2024cardinal} show that finding a fractional matching that is EF and PO is PPAD-complete, and that the MNW solution -- which can be approximated in polynomial time -- achieves $2$-EF. We recover this result as a direct corollary of \Cref{thm:mnw-pf} (MNW $\implies 1$-\pef) and that $1$-\pef implies $2$-EF in the one-sided matching setting. To observe the latter implication, suppose by contradiction that a $1$-\pef matching is not $2$-EF, and the pair of agents $i,j$ witness a violation of $2$-EF. By swapping the bundles of $i$ and $j$, we obtain another matching that the envious agent more than doubles their utilities without hurting anyone but the envied agent, thus failing $1$-\pef. The key observation here is that while we are not allowed to take the union of the allocations to agents $i$ and $j$, and give it entirely to agent $i$ (as we did in cake-cutting), we are allowed to swap the allocations to agents $i$ and $j$, and this is sufficient for $1$-\pef to imply $2$-EF. 

\paragraph{Novel fair division implications.} \Cref{thm:mnw-pf} is connected to a series of works that establish the existence of EF+PO allocations in allocation of homogeneous divisible goods with strongly monotone preferences and a different additional condition: convex preferences~\cite{Var74}, a unique allocation inducing any given (weakly) PO utility vector~\cite{Var74}, the set of allocations inducing any given PO utility vector being convex~\cite{svensson1983existence}, or the last set being a contractible space~\cite{diamantaras1992equity}. Our result provides the existence of a \pef+PO (in fact, a proportionally fair) allocation under a different condition of compactness and upper convexity of utility set $\U$. 

\citet{ray2022fairness} show that a proportionally fair (PF) outcome exists under three conditions: $\O$ is non-empty, compact and convex; each $u_i$ is continuous; and the set of outcomes maximizing any non-negative weighted welfare, $\argmax_{o \in \O} \sum_{i \in N} w_i u_i(o)$, is convex. This is one of the most general PF existence result known. \Cref{thm:mnw-pf} is not subsumed by their result, and points to the possibility of an even more general result for the existence of a PF outcome. 

\citet{Dall01} studies a model that generalizes even cake-cutting, in which each point on the cake (represented as $[0,1]$) can be ``divided'' between the agents, each agent receiving a fraction of the point. Formally, an allocation is given by $\mathbf{\phi} = (\phi_i)_{i \in N}$ with $\phi_i : [0,1] \to [0,1]$ for each $i \in N$ and $\sum_{i \in N} \phi_i(x)=1$ for all $x \in [0,1]$; here, $\phi_i(x)$ denotes the fraction of point $x$ allocated to agent $i$. The utility to agent $i$ is given by $\int_x \phi_i(x) f_i(x) \dd x$, where $f_i$ is a density function of agent $i$. \citet{Dall01} proves that the utility set in this model is compact and convex, allowing \Cref{thm:mnw-pf} to go through to establish the existence of an EF+PO allocation. 

\citet{caragiannis2022beyond} study a model of allocating homogeneous divisible goods, but where the utility function of an agent is weakly monotone over the fraction of a good allocated to her and only additive across the goods. They consider randomized allocations by evaluating the expected utility of the agents. Instead of EF+PO, they seek an allocation satisfying the weaker guarantee of EF and Pareto optimality within the set of EF allocations --- a combination which is trivially guaranteed to be satisfiable. When the agent utilities are concave over the fraction of a good allocated, compactness and convexity of the utility set $\U$ follows and \Cref{thm:mnw-pf} implies the existence of a \emph{deterministic} \pef+PO allocation. \pef does not necessarily imply EF beyond fully additive utilities, so this is another relaxation of EF+PO, but one that is apriori not trivially satisfiable. 

\citet{cole2021existence} study a resource allocation model where an allocation is a lottery over a given finite set of deterministic allocations. Agents have arbitrary utilities for deterministic allocations and calculate expected utilities for a randomized allocation. They show the existence of an EF+PO allocation under a swappability condition on the set of deterministic allocations. \Cref{thm:mnw-pf} implies the existence of a \pef+PO allocation (which is an incomparable guarantee to EF+PO) in their model without any swappability condition. 

\paragraph{Envy in public outcomes.} Finally, our work significantly expands the reach of envy-freeness (or more specifically, the idea of a pairwise notion of individual fairness) to a much larger set of collective decision making domains. In \Cref{sec:computation}, we point out an open question of computing a \pef outcome in voting (where $n$ voters express cardinal preferences over $m$ candidates and a lottery over the candidates is selected) in polynomial time. One can also now explore \pef in domains where the concept of envy has not been explored (but other fairness notions have been), such as clustering~\cite{CFLM19,MS20}, classification~\cite{HMS20}, and federated learning~\cite{ray2022fairness}. 

\section{Additional Experiments}\label{app:more-experiments}

\newcommand{\harmonic}{\operatorname{H}}
For the experiments in the approval-based participatory budgeting setting, in \Cref{fig:variants_nash}, we compare the Nash welfare-maximizing solution to its two close approximations:
\begin{itemize}
    \item Proportional Approval Voting Rule (PAV), which selects the set of projects $O$ that maximizes $\sum_{i} \harmonic(u_i(O))$ where $\harmonic(0) = 0$ and $\harmonic(k) = \sum_{k' = 1}^k \frac{1}{k'}$ for $k \in \N$ is the $k$th harmonic number.
    \item Maximum Smoothed Nash Welfare Rule (Smooth-Nash), which selects the set of projects  that maximizes $\sum_i \log(u_i(O) + 1)$. 
\end{itemize}

\begin{figure}[t]
    \centering
    \begin{subfigure}{0.45\textwidth}
        \centering
        \includegraphics[width=\textwidth]{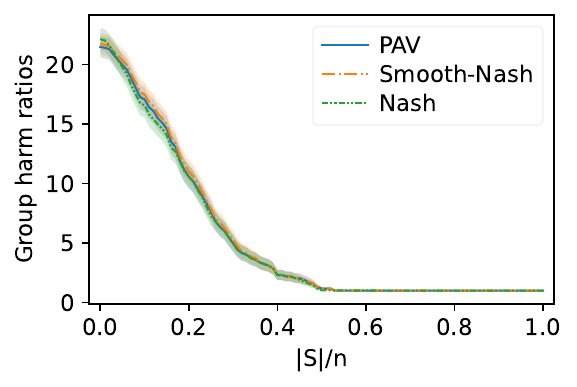}
        \caption{Normalized public group envy ratio as a function of the minimum size of the envious group of agents $S$.}
    \end{subfigure}
    \qquad
    \begin{subfigure}{0.45\textwidth}
        \centering
        \includegraphics[width=\textwidth]{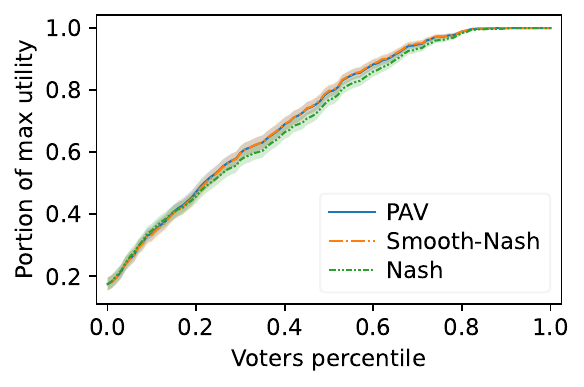}
        \caption{percentiles of voters proportional utilities (fraction of ones maximum feasible utility) averaged over all PB instances.}
    \end{subfigure}
    \caption{Comparison among the Nash welfare-maximizing solution (Nash), Proportional Approval Voting (PAV), and the maximum smoothed Nash nash solution (Smooth-Nash), based on participatory budgeting datasets with approval utilities.}
    \label{fig:variants_nash}
\end{figure}

All three rules achieve similar results. 
\end{document}